\documentclass[11pt,letterpaper]{article}
\usepackage[T1]{fontenc}
\usepackage{mathpazo}
\usepackage{fullpage}
\usepackage{amsmath,amssymb,amsthm,amsfonts}
\usepackage{array}
\usepackage{booktabs}
\usepackage{float}
\usepackage[dvipsnames]{xcolor}
\usepackage{microtype}
\usepackage{tabularx}
\usepackage[style=alphabetic,natbib=true,maxalphanames=9,minalphanames=3,maxcitenames=9,mincitenames=3,maxbibnames=99]{biblatex}
\usepackage{enumitem}
\usepackage{listings}
\usepackage[colorlinks=true,allcolors=magenta]{hyperref}
\usepackage[nameinlink,capitalize,noabbrev]{cleveref}
\usepackage{mathtools}
\usepackage{pifont}
\usepackage{thm-restate}
\usepackage{nicefrac}
\usepackage{tikz}
\usepackage{xspace}

\usepackage[ruled]{algorithm2e} 

\SetAlFnt{\small}
\SetAlCapFnt{\small}
\SetAlCapNameFnt{\small}
\SetAlCapHSkip{0pt}
\IncMargin{-\parindent}
\crefname{algocf}{algorithm}{algorithms}
\Crefname{algocf}{Algorithm}{Algorithms}

\newtheorem{theorem}{Theorem}

\newtheorem{lemma}{Lemma}
\newtheorem{open}{Open Question}
\theoremstyle{definition}
\newtheorem{definition}{Definition}

\theoremstyle{remark}
\newtheorem{remark}{Remark}

\usepackage{tcolorbox}

\renewcommand{\ge}{\geqslant}

\renewcommand{\le}{\leqslant}

\DeclarePairedDelimiter{\set}{\{}{\}}

\DeclarePairedDelimiter{\floor}{\lfloor}{\rfloor}

\newcommand{\N}{\mathbb{N}}

\usepackage{crossreftools}
\AddToHook{cmd/appendix/before}{%
	\crefalias{section}{appendix}%
	\crefalias{subsection}{appendix}
}

\usepackage{changepage}
\usepackage{caption}
\newlength{\arxivpairedheight}
\usepackage{todonotes}

\renewcommand{\epsilon}{\varepsilon}

\DeclareMathOperator{\argmax}{\arg\max}
\newcommand{\myparagraph}[1]{\medskip\noindent\textbf{#1}~}

\newcommand{\calI}{\mathcal{I}}
\newcommand{\calA}{\mathcal{A}}
\newcommand{\calC}{\mathcal{C}}
\newcommand{\pref}{\succ}
\newcommand{\prefp}{\vec{\succ}}
\newcommand{\opt}{\operatorname{OPT}}
\newcommand{\comp}{\operatorname{CR}}
\newcommand{\vs}{\vec{s}}
\newcommand{\kapp}{k\text{-approval}}

\newcommand{\lb}{\operatorname{LB}}
\newcommand{\ub}{\operatorname{UB}}
\newcommand{\level}{\textsc{Level}}
\newcommand{\levelpr}{\textsc{LevelPruning}}
\newcommand{\multiscale}{\textsc{MultiScaleScore}}
\newcommand{\scoretest}{\textsc{ScoreTest}}
\newcommand{\polis}{\href{https://pol.is/home2}{Polis}\xspace}
\newcommand{\remesh}{\href{https://www.remesh.ai/}{Remesh}\xspace}
\newcommand{\preflib}{\href{https://preflib.org/}{PrefLib}}
\newcommand{\score}{\operatorname{score}}
\newcommand{\cost}{\operatorname{cost}}

\lstdefinestyle{pseudocode}{
    basicstyle=\ttfamily\small,
    columns=fullflexible,
    frame=single,
    numbers=left,
    numberstyle=\scriptsize,
    numbersep=8pt,
    xleftmargin=1.5em,
    framexleftmargin=1.2em,
    mathescape=true,
    showstringspaces=false,
    keepspaces=true,
    morekeywords={procedure,while,do,end,choose,query,update,for,if,return}
}

\newenvironment{openunnumbered}{\begin{tcolorbox}[colback=gray!10, colframe=black, rounded corners]\textbf{Open Question:}}{\end{tcolorbox}}

\title{The Art of Calling the Winner by Asking Just Enough Questions:\\Competitive Preference Elicitation with Next-Best Queries}
\author{Nisarg Shah\\University of Toronto\\\texttt{nisarg@cs.toronto.edu} \and Ziqi Yu\\University of Toronto\\\texttt{ziqiyu@cs.toronto.edu}}
\date{}

\begin{document}
\maketitle

\begin{abstract}
  We study active elicitation of agent preferences for collectively choosing among $m$ alternatives using prominent voting rules. We focus on the next-best query model, in which an agent responds to a query by revealing their next favorite alternative, and measure the competitive ratio, which is the worst-case ratio between the number of queries made by the active elicitation algorithm and the minimum number of queries needed to reveal the winning alternative(s) in hindsight. We show that sublinear competitive ratios are achievable for many positional scoring rules, whereas every Condorcet-consistent rule has competitive ratio linear in $m$. For Borda count, we develop two complementary techniques: level-wise pruning, whose analysis extends to general concave scoring rules, and multi-scale score thresholding, which gives an $O(\sqrt m)$ worst-case guarantee for Borda. We also demonstrate strong empirical performance of level-wise pruning on real data.
\end{abstract}

\section{Introduction}\label{sec:intro}

Voting is a fundamental primitive for multiagent decision-making. Beyond political elections, voting mechanisms are routinely invoked to aggregate preferences in human computation systems~\cite{MPC13}, to combine predictions in ensemble systems~\cite{pennock2000normative}, and increasingly to arbitrate among AI agents in multi-LLM architectures~\cite{zhao2024electoral}. In all these settings, a group must collectively select a single alternative from a set of $m$ alternatives based on the ranked preferences of $n$ agents.

Social choice theory provides a rich menu of voting rules, each justified by compelling normative principles~\cite{BCEL+16}. A central modeling assumption underlying this literature is that the complete preference ranking of every agent over the alternatives is available upfront. This assumption is reasonable when $m$ is small, but unrealistic in modern decision-making environments that select from a large pool of alternatives: for example, deliberation platforms such as \polis and \remesh induce voting over thousands of proposed statements~\cite{SBES+21,KSIO23,KTAA+25} and generative systems such as the Habermas machine~\cite{tessler2024ai} find a common ground by searching over an astronomically large space of statements.

This tension has motivated a growing body of work on voting with limited preference information. Given limited information, prior work studies possible and necessary winners---alternatives that win under some and all completions of the available preference information, respectively---under common voting rules~\cite{konczak2005voting,xia2011determining}. When such information can be actively elicited, prior work aims to make just enough queries to reveal the winners and measures the \emph{communication complexity}, i.e., the worst-case number of queries needed~\cite{conitzer2005communication,service2012communication}.


However, most interesting voting rules demand eliciting the entire preference profile in the worst case, unless one makes structural assumptions on the preferences~\cite{conitzer2007eliciting}. Arguably, making many queries is less problematic on \emph{hard instances} in which fewer queries could not have revealed the winners, but more problematic on \emph{easy instances} in which a small number of queries could have revealed the winners in hindsight. This motivates the study of \emph{competitive ratio}~\cite{sleator1985amortized,borodin2005online}, which is the worst-case ratio between the number of queries made by the algorithm and the minimum possible number of queries needed to reveal the winners in hindsight.

\citet{HMSK21} are the first to adopt this approach to social choice. They seek a desirable one-to-one matching of $n$ agents to $n$ objects in the next-best query model, in which an algorithm is allowed to query an agent's $k^{\text{th}}$ most preferred object only after querying her first $k-1$ favorite objects. This model captures sequential top-$k$ interfaces, where the system gradually asks each voter to extend her revealed prefix by one additional alternative, and stops once the winner can be certified. Adapting this to voting, we initiate the study of the \emph{competitive ratio for winner determination under prominent voting rules in the next-best query model}.

With $n$ voters and $m$ alternatives, the competitive ratio is trivially upper bounded by $O(m)$ for most reasonable voting rules: the winners can certainly be determined by eliciting complete preference rankings in $n \cdot m$ queries, whereas any rule that reduces to a simple majority in case of only two alternatives---this is the case for almost all prominent rules---requires making at least one query to at least $\nicefrac{n}{2}$ voters. This leads to our main research question:

\begin{quote}
    \emph{Which voting rules admit an elicitation algorithm with competitive ratio sublinear in the number of alternatives?}
\end{quote}

\subsection{Our Results}\label{sec:our-results}
In \Cref{sec:plurality-approval}, we warm up with a simple analysis proving that plurality, which selects the alternatives that are the top choice of the maximum number of voters, admits an optimal competitive ratio of $2$. This is achieved by the simple rule that elicits all voters' top choices and returns the plurality winners. Then, we extend this analysis to the broader class of $k$-approval rules, which select the alternatives that are among the top $k$ choices of the maximum number of voters for a given value of $k$; note that $k=1$ yields plurality, whereas $k=m-1$ is known as veto. Once again, simply eliciting the top $k$ choices of all the voters and returning the $k$-approval winners turns out to be the best one can do with a competitive ratio of $k$ for $k \in \set{2,\ldots,m-1}$. This is sublinear when $k \in o(m)$.

In \Cref{sec:borda}, we consider Borda count and develop two complementary elicitation algorithms. The first algorithm, $\levelpr$, elicits rankings level-by-level while pruning redundant voters. Its analysis extends to general concave scoring rules; for Borda, it gives an $O(m^{2/3})$ competitive ratio, while a separate lower-bound construction shows that this particular algorithm can incur $\Omega(\sqrt m)$. The second algorithm, $\multiscale$, applies score-upper-bound tests at a sequence of dyadic thresholds and achieves an $O(\sqrt m)$ competitive ratio for Borda. Thus, the two results are complementary in scope: $\levelpr$ applies beyond Borda, whereas $\multiscale$ gives the stronger worst-case asymptotic guarantee for Borda. We also establish a universal lower bound of $3$ against all Borda elicitation algorithms.

In \Cref{sec:condorcet}, we show that no Condorcet-consistent voting rule admits a sublinear competitive ratio in $m$, indicating that this family of voting rules may demand more elicitation than the aforementioned positional scoring rules.

In \Cref{sec:experiments}, we conduct experiments with real datasets from \preflib~\cite{MW13}, which show that our elicitation algorithm achieves even better competitive ratios in practice than the theoretical bounds indicate.

All proofs omitted from the main text appear in the appendix.


\subsection{Related Work}\label{sec:related-work}

\paragraph{Competitive analysis.} While competitive-ratio analysis is standard in the online algorithms literature~\cite{sleator1985amortized,borodin2005online}, its application to active preference elicitation---where the algorithm actively chooses which information to query rather than passively receiving it---is much more recent (see, e.g., \citet{PriceZhao2023}). \citet{HMSK21} introduce the next-best query model and study the competitive ratio for identifying optimal matchings of $n$ agents to $n$ objects. Specifically, they establish a tight bound of $\Theta(\sqrt{n})$ for the criterion of Pareto optimality, while establishing only a lower bound of approximately $\nicefrac{4}{3}$ for a different criterion called rank maximality; \citet{peters2022online} revises the latter to an optimal bound of $\nicefrac{3}{2}$. To the best of our knowledge, our work is the first to adopt this framework to winner determination in voting, where the need to pick a global winner for all $n$ agents rather than matching each of them to a distinct object requires novel algorithmic strategies.

Competitive analysis has also been applied in similar settings. \citet{oren2014online} apply it to an online social choice problem, where the preferences arrive in an online fashion. \citet{chen2018optimal} apply it to analyze algorithms which can extract top $K$ alternatives using noisy pairwise comparisons, with applications to recommender systems.

\paragraph{Partial information.} In social choice theory, decision-making under the uncertainty induced by partial preference information (either given or actively elicited) is well-studied. Beyond the winner determination problem discussed above, it has also been applied to design novel voting rules that optimize welfare~\cite{borodin2022distortion,kempe2020communication} or fairness~\cite{halpern2026representation,ebadian2024metric}.

\section{Model}\label{sec:model}

For $k \in \N$, define $[k] \triangleq \set{1,\ldots,k}$.

\paragraph{Problem instances.} Let $V = [n]$ be a set of $n$ \emph{voters} and $A$ be a set of $m$ \emph{alternatives}. Each voter $v \in V$ has a \emph{preference ranking} (strict total order) $\pref_v$ over the alternatives in $A$. We use $\pi_{\pref_v}(a)$ to denote the rank of alternative $a$ in $\pref_v$, and $\pi_{\pref_v}^{-1}(k)$ to denote the alternative ranked $k$-th in $\pref_v$. For brevity, after this definition we write $\pi_v$ instead of $\pi_{\pref_v}$. We call $\prefp = (\pref_v)_{v \in V}$ the \emph{preference profile}. An instance is given by the tuple $I = (V,A,\prefp)$; we drop $I$ from the notation whenever it is clear from the context. Let $\calI$ be the set of all instances over a fixed set of alternatives $A$ (and any finite set of voters $V$). Effectively, we are fixing the number of alternatives $m$ but allowing any number of voters $n \in \N$.

\paragraph{Voting rules.} A \emph{voting rule} $f : \calI \to 2^A$ maps each instance to a set of (tied) alternatives, termed \emph{winners}. We study the following prominent voting rules.

\begin{itemize}
    \item \emph{Positional scoring rules.} A \emph{scoring vector} $\vs = (s_1, \ldots, s_m)$, with $s_1 \ge s_2 \ge \ldots \ge s_m \ge 0$ and $s_1 > s_m$, awards a score of $s_r$ to each alternative each time it appears in position $r \in [m]$. The corresponding positional scoring rule selects the set of alternatives with the highest total scores, i.e., $\argmax_{a \in A} \sum_{v \in V} s_{\pi_v(a)}$.
    \begin{itemize}
        \item For $k \in [m-1]$, the \emph{$k$-approval rule} ($f_{\kapp}$) awards one point to each of a voter's top $k$ alternatives and zero points to the rest. \emph{Plurality} and \emph{veto} are the special cases $k=1$ and $k=m-1$, respectively.
        \item \emph{Borda count} is given by the scoring vector $(m-1,m-2,\ldots,0)$.
        \item The \emph{harmonic rule} is given by the scoring vector $(1,\nicefrac{1}{2},\ldots,\nicefrac{1}{m})$.
    \end{itemize}
    \item \emph{Condorcet-consistent rules.} Let $n_{a \succ b}(I)$ denote the number of voters who prefer alternative $a$ to $b$ in an instance $I$. Alternative $a$ is a \emph{Condorcet winner} in instance $I$ if $n_{a \succ b}(I) > \nicefrac{n}{2}$ for all $b \in A \setminus \set{a}$. A voting rule $f$ is \emph{Condorcet-consistent} if, on every instance $I$ in which some alternative $a$ is a Condorcet winner, it returns $f(I) = \set{a}$. Two examples are Copeland's rule and the minimax rule.
    \begin{itemize}
        \item The Copeland score of an alternative $a$ in $I$ is obtained by adding $1$ for every alternative $b$ it wins against by pairwise majority ($n_{a \succ b}(I) > n_{b \succ a}(I)$) and $0.5$ for every alternative $b$ it is tied against ($n_{a \succ b}(I) = n_{b \succ a}(I)$). Copeland's rule selects the set of alternatives with the highest Copeland score.
        \item \emph{Minimax rule.} The minimax score of an alternative $a$ in instance $I$ is
        \[
        \min_{b \in A\setminus\set{a}} \bigl(n_{a \succ b}(I)-n_{b \succ a}(I)\bigr).
        \]
        The minimax rule selects the set of alternatives with the highest minimax score.
    \end{itemize}
\end{itemize}

\paragraph{Sequential queries.} A \emph{query sequence} $Q$ is a finite ordered list $((v_t,k_t))_{t \ge 1}$, where $v_t \in V$ and $k_t \in [m]$. Let $|Q|$ denote the length of the list, i.e., the total number of queries. For $t \in [|Q|]$, let $Q[t] \triangleq (v_t,k_t)$ be the $t$-th query, and $Q[1{:}t]$ be the sequence of first $t$ queries; by convention, $Q[1{:}0]$ and its response sequence are empty. Query $(v_t,k_t)$ elicits the $k_t$-th alternative in the preference ranking of voter $v_t$, denoted $R((v_t,k_t),I) \triangleq \pi_{v_t}^{-1}(k_t)$. With slight abuse of notation, let $R(Q[1{:}t],I) \triangleq (R(Q[1],I),\ldots,R(Q[t],I))$ be the responses to the first $t$ queries.

\paragraph{Consistent instances and necessary winners.} Given a query sequence $Q$ and an instance $I=(V,A,\prefp)$, let $\calC(Q,I)$ denote the set of all \emph{consistent instances} $I'=(V,A,\prefp')$ on the same voter and alternative sets that would have yielded the same responses as $I$, i.e., for which $R(Q,I') = R(Q,I)$. We say that alternative $a$ is a \emph{necessary winner} if $a \in f(I')$ for all $I' \in \calC(Q,I)$.

\paragraph{Next-best query model.} A \emph{next-best query sequence} $Q$ elicits any voter's preference ranking from top to bottom, i.e., it elicits any voter $v$'s any $k$-th favorite alternative only after eliciting her top $k-1$ favorite alternatives: formally, for all $t \in [|Q|]$, $v \in V$, and $k \in [m]$, if $Q[t] = (v,k)$, then for all $k' \in [k-1]$, there exists $t' \in [t-1]$ with $Q[t'] = (v,k')$. Let $d(v,Q)$ denote the number of queries posed to $v$ under $Q$, which is the length of the elicited prefix of her preference ranking.

At any intermediate stage, the elicitation algorithm observes only a partial preference profile: for each voter, a prefix of her ranking has been revealed, while the remaining alternatives are still unelicited. Thus, the algorithm must reason under uncertainty about the unobserved parts of the profile and decide which queries are necessary to certify the winner.

\paragraph{Elicitation in the next-best query model.} An \emph{elicitation algorithm} $\calA$ for voting rule $f$ maps every instance $I \in \calI$ to a query sequence $Q \triangleq \calA(I) = ((v_1,k_1),(v_2,k_2),\ldots)$ of some length $\ell$ satisfying three conditions:
\begin{enumerate}
    \item[{[C1]}] (Winner Determination) The query responses must determine the set of winners, i.e., for every $I' \in \calC(Q,I)$, we must have $f(I') = f(I)$.
    \item[{[C2]}] (Next-Best) $Q$ must be a next-best query sequence.
    \item[{[C3]}] (Online Elicitation) The next query must depend only on the elicited responses, i.e., for any $t \in \set{0,\ldots,\ell-1}$ and any consistent instance $I' \in \calC(Q[1{:}t],I)$, we must have $|\calA(I')| \ge t+1$ and $\calA(I')[t+1] = \calA(I)[t+1]$.
\end{enumerate}


\paragraph{Optimal queries and competitive ratio.} Let $Q^\star(I)$ be any shortest sequence satisfying conditions [C1] and [C2] above, but not necessarily [C3], and let $\opt(I)\triangleq |Q^\star(I)|$ be its cost. Note that in the absence of [C3], $Q^\star(I)$ is characterized by the prefix length $d(v,Q^\star(I))$ elicited from each voter $v$. While $\opt(I)$ is the smallest number of queries that reveal the winners under $I$, an online elicitation algorithm $\calA$ oblivious of the instance may end up making more queries before deducing the winners; its \emph{competitive ratio} is given by:
\[
\comp(\calA) = \sup_{I \in \calI} \frac{|\calA(I)|}{\opt(I)}.
\]

All our upper bounds determine the complete set of tied winners. Our lower bounds use the same objective unless explicitly stated otherwise.

\section{Warm Up: Plurality and $k$-Approval}\label{sec:plurality-approval}

Plurality is arguably the most widely used voting rule. Recall that plurality winners are simply the alternatives that are the top choices of the highest number of voters. Consequently, eliciting the top choices of all the voters (in any order) suffices to reveal them with $n$ queries.  \emph{What is the competitive ratio of this algorithm? Is it the best possible?} The following simple analysis shows that it is essentially $2$ and indeed the best possible.

\begin{restatable}{theorem}{thmPlurality}\label{thm:plurality}
    The algorithm that elicits the top choices of all the voters and returns the plurality winners has competitive ratio $2$. No online algorithm for the plurality rule in the next-best query model has competitive ratio less than $2$.
\end{restatable}
\begin{proof}
Let us establish an upper bound on the competitive ratio of the algorithm in question and the same lower bound on the competitive ratio of all online algorithms.

\emph{Upper bound.} Because the algorithm makes $n$ queries on every instance, it is sufficient to establish that $\opt(I) \ge \nicefrac{n}{2}$ for all instances $I$. Suppose for contradiction that there is an instance $I$ in which one can conclude that alternative $a$ is a necessary plurality winner after making fewer than $\nicefrac{n}{2}$ queries. Then, more than $\nicefrac{n}{2}$ voters would have received no queries. If all these voters rank another alternative $b$ as their top choice, it would make $b$ the only plurality winner, contradicting the assumption that $a$ is a necessary plurality winner. This implies that $\opt(I) \ge \nicefrac{n}{2}$, so the competitive ratio of the algorithm is at most $2$.

\emph{Lower bound.} Consider any online elicitation algorithm $\calA$. Fix any two alternatives $a,b \in A$ and construct an instance adversarially as follows. The first $\floor{\nicefrac{n}{2}}$ voters $v$ receiving a query of $(v,1)$ (for their top choice) respond with alternative $a$. The next $\floor{\nicefrac{n}{2}}-1$ voters $v$ receiving a query of $(v,1)$ (for their top choice) respond with alternative $b$. Any voter $v$ receiving a query of $(v,k)$ for $k > 1$ responds with an arbitrary consistent response (i.e., with any alternative that they have not revealed in response to a query $(v,k')$ for any $k' < k$). It is easy to see that this information is not sufficient to deduce the plurality winner as the gap between the plurality scores of $a$ and $b$ is $1$ and there are at least two voters who have not revealed their top choices. Thus, the algorithm must make at least $2\floor{\nicefrac{n}{2}}-1$ queries on all instances consistent with the revealed responses thus far. In contrast, consider the instance $I$ in this family, in which $a$ is the top choice of $\floor{\nicefrac{n}{2}}+1$ voters; revealing the top choices of these voters would suffice to deduce $a$ as a necessary plurality winner in $I$. Hence, the competitive ratio of $\calA$ is
\[
\comp(\calA) \ge \lim_{n \to \infty} \frac{2\floor{\nicefrac{n}{2}}-1}{\floor{\nicefrac{n}{2}}+1} = 2,
\]
as desired.
\end{proof}

Recall that plurality is the $1$-approval rule. Let us move on to the $k$-approval rule, for a given $k \in \set{2,3,\ldots,m-1}$, which selects the alternatives that appear in the top $k$ positions of the maximum number of voters. With careful analysis, the argument for plurality generalizes: simply eliciting the top $k$ positions of all $n$ voters and returning the set of $k$-approval winners yields an optimal competitive ratio of $k$. Note that this optimal ratio does not change when moving from $k=1$ to $k=2$, but increases linearly afterwards.

\begin{restatable}{theorem}{thmkApproval}\label{thm:k-approval}
    For $k \in \set{2,3,\ldots,m-1}$, the algorithm that elicits the top $k$ choices of all the voters and returns the $k$-approval winners has competitive ratio $k$, and no online elicitation algorithm for the $k$-approval rule in the next-best query model has competitive ratio less than $k$.
\end{restatable}

\section{Borda Count (and Concave Scoring Rules)}\label{sec:borda}

Let us now move on to Borda count, a widely used and one of the most classical voting rules. Suddenly, we discover that the analysis gets a lot more complex than \Cref{thm:plurality}.

First, we can show a slightly improved lower bound for Borda count compared to the bound for plurality in \Cref{thm:plurality}, although this requires a significantly more involved construction.
\begin{restatable}{theorem}{thmBordaLower}\label{thm:borda-lower}
    The competitive ratio of any online elicitation algorithm for Borda count is at least $3$.
\end{restatable}

How do we achieve a low competitive ratio for Borda count? For plurality, recall that we simply elicited the top choices of all voters. This information does not suffice to identify necessary Borda count winners.

What if we elicit the entire preference profile with $n \cdot m$ queries? Then, we face a terrible competitive ratio of $m$ due to the simple instance in which all voters happen to rank the same alternative first and $O(n)$ queries would suffice to reveal it as the unique Borda count winner.

This suggests the following refined approach. Elicit voters' preference rankings \emph{level by level}, i.e., first everyone's top choices,\footnote{Voters in each level are elicited in an arbitrary order.} then everyone's second-best choices, and so on, and, crucially, \emph{stop} as soon as the set of Borda winners is identified. This algorithm, which we term $\level$, works well on the simple instance above, but still incurs the same $\Omega(m)$ competitive ratio.

\begin{restatable}{proposition}{propLevelBad}\label{prop:level-bad}
    The competitive ratio of $\level$ for Borda count is $\Omega(m)$.
\end{restatable}
\begin{proof}
    To see this, consider an instance with $n=2p+1$ voters, and a set $A$ of $m$ alternatives. Fix two special alternatives $a,b \in A$. Voters $1,\ldots,p$ rank $a$ first, $b$ second, and the remaining alternatives arbitrarily. Voters $p+1,\ldots,2p$ rank $b$ first, $a$ second, and the remaining alternatives arbitrarily. And voter $2p+1$ ranks $b$ \emph{last}, $a$ second-to-last, and the remaining alternatives arbitrarily in the first $m-2$ positions. Note that $a$ is the unique Borda count winner, but $\level$ will notice a tie between $a$ and $b$ after eliciting two levels in everyone's rankings, and will need to continue eliciting till level $m-1$ before it sees $a$ in the ranking of voter $2p+1$.

    Hence, $\level$ will terminate after $n(m-1)$ queries, whereas $2n+(m-2)$ queries in which all voters reveal their top two choices and voter $2p+1$ further reveals ranks $3,\ldots,m-1$ would suffice to establish $a$ as the unique Borda count winner, implying that the competitive ratio of $\level$ is at least
    \[
    \lim_{n \to \infty} \frac{n(m-1)}{2n+(m-2)} = \frac{m-1}{2} \in \Omega(m).\qedhere
    \]
\end{proof}

In the bad instance in the proof of \Cref{prop:level-bad}, $\level$ performs terribly because after eliciting the first two levels, we can already narrow down to $a$ and $b$ as the only \emph{possible winners} under Borda count. Since voters $1,\ldots,2p$ have already revealed both $a$ and $b$, it would make sense to \emph{prune} them and not query them further. Yet, $\level$ queries all voters at each level, which leads to the bad competitive ratio.

Based on this insight, we next consider an algorithm, termed $\levelpr$ and presented as \Cref{alg:level-pruning}, which mimics $\level$, but after each level, prunes (and makes no further queries to) voters who have revealed all alternatives that can plausibly be winners in the end. Although motivated here by Borda count, the algorithm is defined for an arbitrary positional scoring rule through its score lower and upper bounds. The following result characterizes a condition that allows us to rule out some alternatives that cannot possibly be a winner.

\begin{restatable}{lemma}{possibleWinners}\label{lem:possible-winners}
    Given a next-best query sequence $Q$, its responses $R$, and an alternative $a$, define $\lb(a;Q,R)$ and $\ub(a;Q,R)$ to be the lowest and highest Borda scores of $a$, respectively, across all instances which yield responses $R$ to query sequence $Q$. Then, alternative $a$ is a possible Borda winner under some completion consistent with $R$ only if $\ub(a;Q,R) \ge \max_{b \in A}\lb(b;Q,R)$.
\end{restatable}
\begin{proof}
    If there exists $b \in A$ such that $\ub(a;Q,R) < \lb(b;Q,R)$, then $a$ would have a lower Borda score than $b$ in every completion consistent with $R$. Hence, $a$ cannot be a possible winner under any completion in this case.
\end{proof}

\Cref{lem:possible-winners} motivates $\levelpr$ as follows. After each level, the algorithm computes $P=\{a\in A:\textsf{UB}(a)\ge \max_{b\in A}\textsf{LB}(b)\}$, a superset of the set of alternatives that can still be Borda count winners, and prunes active voters who have already revealed all alternatives in $P$. Since lower bounds can only increase and upper bounds can only decrease as more prefixes are revealed, $P$ can only shrink, so a pruned voter has already revealed every alternative that may remain relevant. When no active voters remain, every alternative in $P$ has an exact score (i.e., $\ub(a)=\lb(a)$ for all $a \in P$) and alternatives outside $P$ cannot be possible winners. Further, all alternatives in $P$ must have the same (exact) score and must all be winners because if $a \in P$ has a lower score than $b \in P$, then $\ub(a) < \lb(b)$ would hold, which would have eliminated $a$ from $P$ in the last update, a contradiction. Hence, in the end, the algorithm simply returns $P$.

\begin{algorithm}[t]
\caption{$\levelpr$ for Concave Scoring Rules}
\label{alg:level-pruning}
\DontPrintSemicolon
\KwIn{Voters $V$, alternatives $A$ with $|A|=m$, and a scoring vector $\vs$}
\KwOut{Winner(s) under the scoring rule induced by $\vs$}

Initialize each voter’s revealed prefix to empty\;
$U \gets V$ \tcp*{active voters}

\For{$d=1,\ldots,m$}{
    Query each active voter $i \in U$ for her $d$-th most favorite alternative \tcp*{this is a next-best query; voters in $U$ have revealed their top $d-1$ alternatives in prior iterations}

    \BlankLine
    \ForEach{$a \in A$}{
        $\textsf{LB}(a) \gets$ minimum possible score of $a$ under $\vs$
        \tcp*{place $a$ last in every unrevealed position}

        $\textsf{UB}(a) \gets$ maximum possible score of $a$ under $\vs$
        \tcp*{place $a$ as high as possible among unrevealed positions}
    }

    $\textsf{bestLB} \gets \max_{b \in A} \textsf{LB}(b)$\;
    $P \gets \{a \in A : \textsf{UB}(a) \ge \textsf{bestLB}\}$ \tcp*{superset of possible winners}
    \BlankLine
    Remove every voter from $U$ who has revealed all alternatives in $P$ \tcp*{prune inactive voters}

    \If(\tcp*[f]{no active voters remain}){$U = \emptyset$}{
        \Return $P$\;
    }
}
\end{algorithm}

The analysis of \levelpr{} extends beyond Borda to scoring vectors whose score losses accelerate with rank. This extension both identifies the structural property used by the algorithm and provides a corrected proof of the Borda upper bound.

\begin{definition}[Concave scoring rule]\label{def:concave-scoring}
A scoring vector $(s_1,\ldots,s_m)$ with $s_1>s_m$ is \emph{concave} if
\[
s_j-s_{j+1}\le s_{j+1}-s_{j+2}
\qquad\text{for every }j\in\{1,\ldots,m-2\}.
\]
Its normalized \emph{top gap} is
$\gamma=(s_1-s_2)/(s_1-s_m)$. For an alternative $a$, its normalized deficit is
$\delta(a)=(ns_1-\score(a))/(n(s_1-s_m))$.
\end{definition}

Positive affine transformations of the scoring vector preserve winners, certificates, the execution of \levelpr, concavity, $\gamma$, and all deficits. We may therefore normalize $s_1=1$ and $s_m=0$. Concavity then implies
\[
s_j-s_{j+1}\ge\gamma,
\qquad
s_j\le1-(j-1)\gamma,
\qquad
1-s_{j+1}\le\frac{j}{m-1}.
\]
In particular, $\gamma\le1/(m-1)$, with equality exactly for Borda. Veto is also concave but has $\gamma=0$; it falls into the full-elicitation branch of \Cref{thm:concave-levelpr}. More directly, veto is $(m-1)$-approval, so \Cref{thm:k-approval} already gives the tight competitive ratio $m-1$, connecting this boundary case to the preceding $k$-approval analysis. The harmonic rule is not concave under \Cref{def:concave-scoring}, because its adjacent score gaps decrease with rank.

\begin{theorem}[Concave scoring rules]\label{thm:concave-levelpr}
Let $(s_1,\ldots,s_m)$ be a concave scoring rule with $m\ge3$ and top gap $\gamma\ge0$, and let $w$ be a winner of deficit $\delta$. Then the cost of \levelpr{} is at most
\[
\begin{cases}
n(m-1),&\gamma=0\text{ or }\delta^2>\gamma,\\[2pt]
\dfrac{7}{4}n\dfrac{\delta^{2/3}}{\gamma^{4/3}}+3n+m,
&\gamma>0,\ \gamma^2\le\delta\text{ and }\delta^2\le\gamma,\\[6pt]
5n+m,&\gamma>0\text{ and }\delta<\gamma^2.
\end{cases}
\]
Consequently, when $\gamma>0$, on every profile with $n\ge m$,
\[
\frac{\cost(\levelpr)}{\opt}
\le
\frac{7}{2(2(m-1))^{2/3}\gamma^{4/3}}
+\frac{1}{\sqrt\gamma}+12
=O\!\left(\frac{1}{m^{2/3}\gamma^{4/3}}\right).
\]
\end{theorem}

\begin{proof}[Proof sketch]
Normalize the scores to $s_1=1$ and $s_m=0$, and fix a winner $w$ with $\score(w)=n(1-\delta)$. Any certificate must satisfy
\[
\opt\ge\max\{n(m-1)\delta,n/2\}.
\]
The first bound follows because $e_i$ queries to voter $i$ can reduce a challenger's upper-bound contribution by at most $e_i/(m-1)$; the second follows because fewer than $n/2$ queries leave too many completely unqueried voters to separate a winner's lower bound from a challenger's upper bound.

When $\gamma=0$, the first branch is the immediate bound obtained by eliciting every ranking through depth $m-1$. For the nontrivial regime $\gamma>0$ and $\delta^2\le\gamma$, choose
\[
\lambda=\min\{(\gamma\delta)^{-1/3},\gamma^{-1}\},
\qquad
\beta=\max\{\lambda\delta,\lambda^{-1}\},
\qquad
d^*=\min\{\lceil\beta/\gamma\rceil,m-1\}.
\]
Phase~1 consists of the iterations through depth $d^*$. Concavity guarantees that the score loss by this depth is at least $\beta$. It follows that at most $n/\lambda$ voters have not made $w$ exact and that $\lb(w)\ge n(1-1/\lambda)$.

For Phase~2, let $\Lambda=\max_b\lb(b)$ and use the potential
\[
\Phi=\sum_{a\in A}(\ub(a)-\Lambda+\gamma)^+.
\]
This potential is bounded by a simple double-counting argument over voters and alternatives. An alternative already outside the possible-winner set contributes at most the $\gamma$ cushion. For an alternative still in that set, every non-exact voter contribution is at most the Phase~1 baseline $L=1-1/\lambda$; hence positive surplus above $L$ can come only from exact voter--alternative pairs. At a fixed voter, these exact alternatives occupy distinct ranks, so their total positive surplus is at most
\[
\sum_{j=1}^{m}(s_j-L)^+.
\]
Concavity gives $s_j-L\le\lambda^{-1}-(j-1)\gamma$, so only the initial ranks with $(j-1)\gamma<1/\lambda$ can contribute positively. Summing this truncated arithmetic progression for each voter, and adding the $m\gamma$ cushions, gives
\[
\Phi\le n\left(\frac{1}{2\lambda^2\gamma}
+\frac{1}{2\lambda}+\frac{\gamma}{8}\right)+m\gamma.
\]
Immediately before each Phase~2 query, consider the current possible-winner set. Either some currently possible winner remains non-exact at the queried voter, in which case its upper bound and hence $\Phi$ fall by at least $\gamma$, or all current possible winners become exact at that voter, in which case the voter is pruned at the end of the iteration. Thus Phase~2 costs at most $\Phi/\gamma+n$. Combining the two phases and substituting the chosen $\lambda$ yields the three cost regimes above. Dividing the main and additive terms separately by the two lower bounds on $\opt$, and then maximizing over the regimes, gives the stated ratio. The complete proof appears in \Cref{app:concave-levelpr-proof}.
\end{proof}

\begin{restatable}{corollary}{thmLevelprGood}\label{thm:levelpr-good}
For Borda count, \levelpr{} has competitive ratio $O(m^{2/3})$.
\end{restatable}

\begin{proof}
First suppose \(n\ge m\). Normalized Borda scores have $\gamma=1/(m-1)$. Substitution in \Cref{thm:concave-levelpr} gives
\[
\frac{\cost(\levelpr)}{\opt}
\le\frac{7}{2^{5/3}}(m-1)^{2/3}+\sqrt{m-1}+12
=O(m^{2/3}).
\]
Now consider a profile \(I\) with \(n<m\), let \(q=\lceil m/n\rceil\), and form \(I^{(q)}\) by replacing every ballot with \(q\) identical copies. The run of \levelpr{} is replicated across the copies, so
\[
\cost(\levelpr,I^{(q)})=q\cost(\levelpr,I).
\]
Repeating an optimal certificate for \(I\) on every copy gives a certificate for \(I^{(q)}\): every completion can be partitioned into \(q\) profiles consistent with the original certificate, all with the same Borda winner set, and summing their scores preserves that set. Hence \(\opt(I^{(q)})\le q\opt(I)\), and therefore
\[
\frac{\cost(\levelpr,I)}{\opt(I)}
\le
\frac{\cost(\levelpr,I^{(q)})}{\opt(I^{(q)})}.
\]
The replicated profile has \(qn\ge m\), so the first part of the proof bounds the right-hand side by \(O(m^{2/3})\). Because \(\calI\) contains every finite electorate and the competitive ratio takes the supremum over \(\calI\), this proves the unrestricted claim.
\end{proof}

It is unclear if our analysis in \Cref{thm:levelpr-good} is tight. Although, we are able to establish the following lower bound, showing that $\levelpr$ cannot achieve a constant competitive ratio matching the universal lower bound of \Cref{thm:borda-lower}.

\begin{restatable}{theorem}{thmLevelprBad}\label{thm:levelpr-bad}
    The competitive ratio of $\levelpr$ for Borda count is $\Omega(\sqrt{m})$.
\end{restatable}

The lower bound in \Cref{thm:levelpr-bad} is specific to \levelpr: it shows a limitation of uniform level-wise pruning, rather than a lower bound against arbitrary Borda elicitation algorithms. We next apply the same score-bound viewpoint at multiple thresholds to obtain a sharper upper bound for Borda count specifically; we are not able to extend a version of this bound to other concave scoring rules.

\subsection{An Improved $O(\sqrt{m})$ Competitive Ratio for Borda Count}\label{sec:multiscale-borda}

We continue to use the score language of \Cref{sec:borda}. Normalize Borda as
\[
s_j=\frac{m-j}{m-1},
\qquad
\score(a)=\sum_{i\in V}s_{\pi_i(a)}.
\]
This is the same normalization used for the Borda specialization of \Cref{thm:concave-levelpr}. At any point, let $d_i\in\{0,\ldots,m-1\}$ be the revealed depth of voter $i$. We say that the score contribution of $a$ is \emph{exact} at voter $i$ if $a$ has been revealed, or if $d_i=m-1$ and $a$ is the unique omitted alternative. Define the per-voter score bounds
\[
\textsf{LB}_i(a):=
\begin{cases}
 s_{\pi_i(a)},&\text{if $a$ is exact at voter $i$},\\
 0,&\text{otherwise},
\end{cases}
\qquad
\textsf{UB}_i(a):=
\begin{cases}
 s_{\pi_i(a)},&\text{if $a$ is exact at voter $i$},\\
 s_{d_i+1},&\text{otherwise},
\end{cases}
\]
and let $\textsf{LB}(a)=\sum_i\textsf{LB}_i(a)$ and $\textsf{UB}(a)=\sum_i\textsf{UB}_i(a)$. Then
\[
\textsf{LB}(a)\le\score(a)\le\textsf{UB}(a).
\]
Moreover, $\textsf{UB}(a)$ is nonincreasing as queries are made, and $a$ has an exact score precisely when its contributions are exact at every voter.

The subroutine $\scoretest(a,T)$ repeatedly queries the least-index voter at which $a$ is not exact until either $\score(a)$ becomes exact or $\textsf{UB}(a)<T$. If $U_0(a)$ is the upper bound when the test starts, then it makes at most
\[
(m-1)(U_0(a)-T)+1+n
\]
queries: an answer different from $a$ decreases $\textsf{UB}(a)$ by $1/(m-1)$, while an answer equal to $a$ can occur at most once per voter. In particular, if $\score(a)\ge T$, the test must make $\score(a)$ exact.

The algorithm considers dyadic scales $g$ up to the largest power of two not exceeding $\sqrt m$. At scale $g$, it first reveals every ballot to depth $2g$, sets the score threshold
\[
T_g:=\frac{n(m-g)}{m-1},
\]
then freezes the set of alternatives whose score upper bound reaches $T_g$, and score-tests every alternative in that set. Freezing and processing the entire set is necessary to recover all co-winners, rather than returning after the first successful test.

\begin{algorithm}[H]
\caption{$\multiscale$ for Borda Count}
\label{alg:multiscale-score}
\DontPrintSemicolon
\KwIn{Voters $V$, alternatives $A$ with $|A|=m$}
\KwOut{Borda winner(s)}

\If{$m\le4$}{
    Extend every ballot to depth $m-1$ and return $\argmax_{a\in A}\score(a)$\;
}
$G\gets2^{\lfloor\log_2\sqrt m\rfloor}$\;
\ForEach{$g\in\{1,2,4,\ldots,G\}$}{
    Extend every ballot whose current depth is below $2g$ to depth $2g$\;
    $T_g\gets n(m-g)/(m-1)$ and $C_g\gets\{a\in A:\textsf{UB}(a)\ge T_g\}$ \tcp*{freeze $C_g$}
    \ForEach{$a\in C_g$ in a fixed order}{
        \If{$\textsf{UB}(a)\ge T_g$}{Run $\scoretest(a,T_g)$\;}
    }
    $E_g\gets\{a\in C_g:\score(a)\text{ is exact and }\score(a)\ge T_g\}$\;
    \If{$E_g\ne\emptyset$}{
        \Return $\argmax_{a\in E_g}\score(a)$\;
    }
}
Extend every ballot to depth $m-1$ and return $\argmax_{a\in A}\score(a)$\;
\end{algorithm}

\begin{restatable}{theorem}{thmMultiScaleBorda}\label{thm:multiscale-borda}
\multiscale{} is a deterministic online next-best elicitation algorithm that determines the complete Borda winner set. Moreover,
\[
\comp(\multiscale)
\le
\begin{cases}
4,&2\le m\le4,\\[2pt]
20\sqrt m,&m\ge5.
\end{cases}
\]
In particular, $\comp(\multiscale)=O(\sqrt m)$.
\end{restatable}

\begin{proof}[Proof sketch]
The first ingredient is a packing bound stated entirely in terms of score upper bounds. Let $\textsf{UB}^{(h)}(a)$ be the upper bound obtained by truncating every ballot at depth $h$. Double counting its surplus above the common unrevealed baseline $ns_{h+1}$ gives
\[
\sum_{a\in A}\left(\textsf{UB}^{(h)}(a)-ns_{h+1}\right)
=\frac{nh(h+1)}{2(m-1)}.
\]
Taking $h=2g$, every $a\in C_g$ contributes at least $T_g-ns_{2g+1}=n(g+1)/(m-1)$ to the left-hand side, and hence $|C_g|\le2g-1$. Moreover, $U_0(a)\le n$ makes each score test cost at most $ng+1$. Thus the total cost through scale $g$ is at most $5ng^2$.

For correctness, after every member of the frozen set $C_g$ has been processed, each alternative outside
\[
E_g=\{a\in C_g:\score(a)\text{ is exact and }\score(a)\ge T_g\}
\]
has $\textsf{UB}(a)<T_g$. Thus, if $E_g$ is nonempty, no alternative outside it can match the maximum exact score attained within $E_g$ in any consistent completion, and the returned set is exactly the complete Borda winner set, including ties.

Finally, let $w$ be a Borda winner and use the same normalized score deficit as in \Cref{def:concave-scoring}:
\[
\delta=1-\frac{\score(w)}n.
\]
Every certificate satisfies
\[
\opt\ge\max\left\{\frac{n(m-1)}m,\;n(m-1)\delta\right\}.
\]
Write $\rho=1+(m-1)\delta$, so that $\score(w)\ge T_g$ exactly when $g\ge\rho$. If $\rho\le G$, let $g$ be the first scale with $g\ge\rho$. Then $w\in C_g$, the algorithm terminates by that scale, and the cost bound $5ng^2$ divided by the lower bound above is at most $20g\le20\sqrt m$ (with the two smallest scales handled by the first lower bound). If $\rho>G$, even full elicitation costs $n(m-1)$, while $\opt\ge n(m-1)\delta=n(\rho-1)>n(G-1)$; since $G>\sqrt m/2$, the ratio is less than $4\sqrt m$. The complete proof appears in \Cref{app:multiscale-proof}.
\end{proof}

\begin{openunnumbered}
Does Borda admit a constant-competitive elicitation algorithm? The universal lower bound is $3$, whereas \Cref{thm:multiscale-borda} gives an $O(\sqrt m)$ upper bound.
\end{openunnumbered}

\section{Condorcet Methods}\label{sec:condorcet}

A popular family of voting rules, disjoint from the family of positional scoring rules~\cite{Fish74}, is that of Condorcet-consistent rules, which includes prominent rules such as Copeland's rule, the minimax rule, Kemeny's rule, ranked pairs, Schulze's method, and Dodgson's method. While we showed that many positional scoring rules admit $o(m)$ competitive ratio, the following result shows that all Condorcet-consistent rules admit $\Omega(m)$ competitive ratio, indicating that determining a Condorcet winner can incur significant cost due to the online nature of the elicitation.


\begin{restatable}{theorem}{thmCondorcet}\label{thm:condorcet}
    The competitive ratio of any elicitation algorithm for any Condorcet-consistent voting rule is $\Omega(m)$.
\end{restatable}
\begin{proof}
    Suppose $n$ and $m$ are even. Consider a set of alternatives $A = \set{w,c} \cup X \cup Y$, where $|X|=|Y|=\nicefrac{m}{2}-1$. We construct an adversarial class of problem instances using \Cref{tab:condorcet} as our guide; whenever $X$ or $Y$ appear in the preference ranking, we substitute an arbitrary order among the alternatives in that set instead.

    \begin{table}[htb!]
    \centering
    \begin{tabular}{crl}
    \toprule
    Voter type & No. of voters & Preference ranking\\
    \midrule
    (T1) &$\nicefrac{n}{2}-1$ voters: &$w \succ c \succ X \succ Y$\\
    (T2) &$2$ voters: & $Y \succ c \succ w \succ X$\\
    (T3) &$2$ voters: & $X \succ c \succ w \succ Y$\\
    (T4) &$\nicefrac{n}{2}-5$ voters: & $X \succ Y \succ c \succ w$\\
    (T5) &$2$ voters: & $X \succ Y \succ w \succ c$\\
    \bottomrule
    \end{tabular}
    \caption{Adversarial problem instances for Condorcet-consistent rules.}
    \label{tab:condorcet}
    \end{table}

    The adversary acts as follows. For the first $\nicefrac{n}{2}-1$ voters receiving their first query, the adversary assigns type (T1) as soon as the first query is received and responds to all their queries as per the table. The same holds for the next $2$ voters receiving their first query, who are assigned type (T2), and the subsequent $2$ voters receiving their first query, who are assigned type (T3). For all the remaining voters, who may be of type (T4) or (T5), the adversary continues revealing $X \succ Y$ for the first $m-2$ positional queries without determining the type. For the first $\nicefrac{n}{2}-5$ of these voters receiving a query at position $m-1$, the adversary assigns type (T4) and reveals alternative $c$ (and responds to a query at position $m$, if ever posed, with alternative $w$). The last $2$ voters in this group receiving a query at position $m-1$ are assigned type (T5) and their queries at positions $m-1$ and $m$ are responded to accordingly.

    First, notice that $w$ is the Condorcet winner: the majority (T1)+(T2) prefer $w \succ X$, the majority (T1)+(T3) prefer $w \succ Y$, and the majority (T1)+(T5) prefer $w \succ c$.

    Next, we argue that any elicitation algorithm for any Condorcet-consistent rule must make at least $(\nicefrac{n}{2}-4) \cdot (m-1)$ queries before terminating (with $w$ as the necessary winner). Suppose for contradiction that it terminates with fewer queries. Consider the group of (T4)+(T5) voters. Because the adversary only assigns type (T5) to the $(\nicefrac{n}{2}-4)^{\text{th}}$ and $(\nicefrac{n}{2}-3)^{\text{rd}}$ voter in this group receiving a query at position $m-1$, the adversary must not have assigned type (T5) to any voters yet. Then, the algorithm cannot rule out the possibility that all voters in the group (T4)+(T5) have preference ranking $X \succ Y \succ c \succ w$. For voters of types (T4) and (T5), the relative order between $c$ and $w$ remains unrevealed at this point, so the adversary can complete their rankings with $c \succ w$ without contradicting any answer given so far. However, in this possibility, $c$ becomes the Condorcet winner: the majorities (T1)+(T2) and (T1)+(T3) prefer $c \succ X$ and $c \succ Y$, respectively, just like $w$, but now the majority $(T2)+(T3)+(T4)+(T5)$ would also prefer $c \succ w$. But then, the algorithm cannot possibly terminate with $w$ as the necessary winner under any Condorcet-consistent voting rule, a contradiction.

    Let us contrast this with the optimal number of queries. Note that $\opt$ can reveal the Condorcet winner $w$ by querying (T1) voters up to rank $1$, (T2) and (T3), and (T5) voters up to rank $m-1$. This amounts to $O(n+m)$ queries.

    Taking the ratio of the $\Omega(nm)$ lower bound on the algorithm and $O(n+m)$ upper bound on $\opt$, followed by the limit of $n \to \infty$, yields the desired $\Omega(m)$ lower bound for even \(m\). For odd \(m\), apply the construction to \(m-1\) alternatives and add one dummy alternative at the bottom of every ranking; this preserves the argument and gives \(\Omega(m-1)=\Omega(m)\).
\end{proof}

\section{Experiments}
\label{sec:experiments}

In this section, we evaluate the empirical performance of our proposed query strategies across various positional scoring rules and Condorcet-consistent rules. The implementation and experiment scripts are available in the
\href{https://github.com/yzqlll0721/NextBest}{NextBest repository}.

\subsection{Experimental Setup}

\paragraph{Datasets.}
We conduct our experiments using real-world preference datasets from \textbf{PrefLib}~\cite{MW13} under the Strict Orders Complete (SOC) category. \Cref{fig:preflib-distribution} depicts the distribution of the number of voters $n$ and the number of alternatives $m$ in these datasets collectively.

\begin{figure}[t]
    \centering
    \begin{minipage}[t]{0.48\linewidth}
        \centering
        \includegraphics[height=\arxivpairedheight,keepaspectratio]{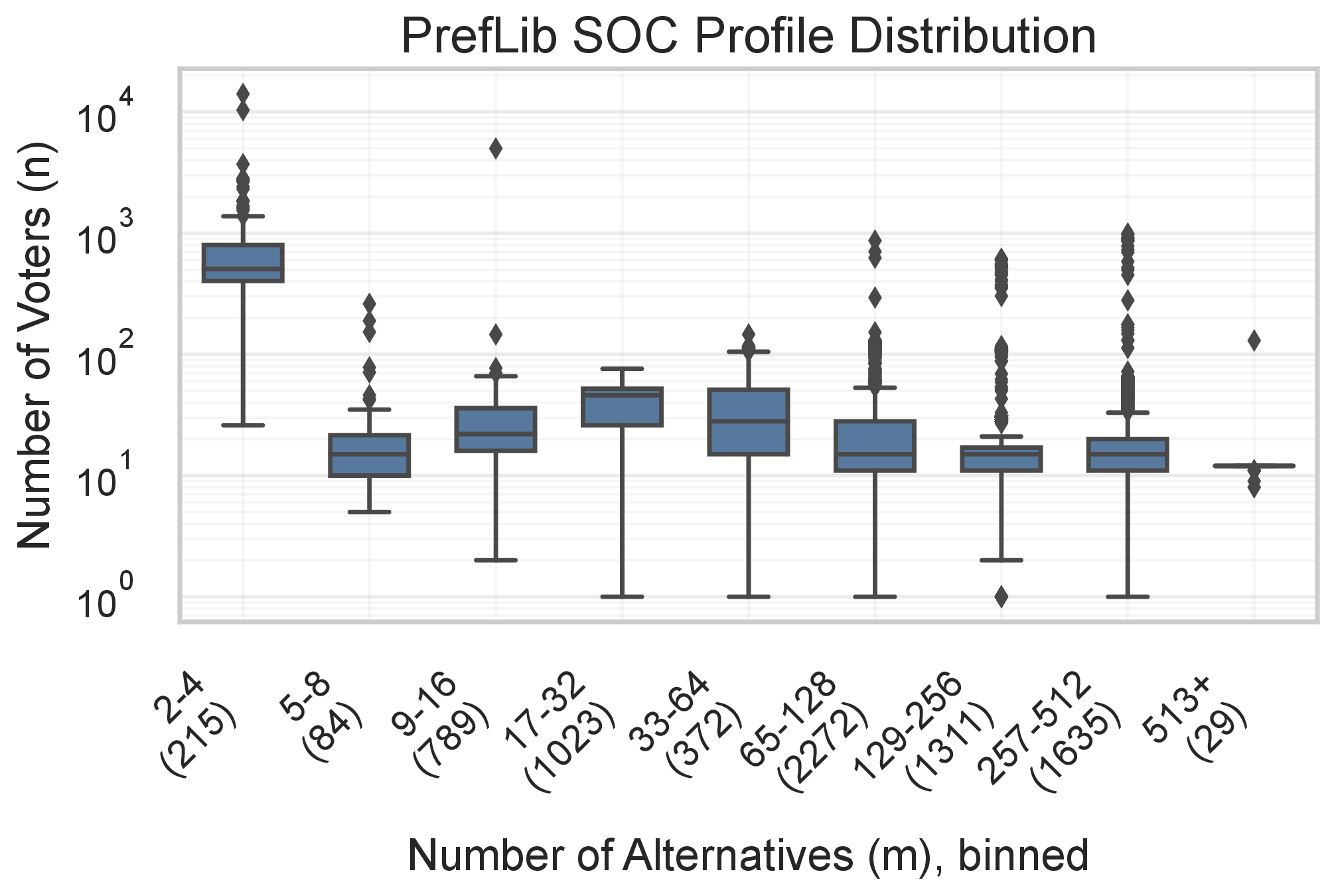}
        \caption{Distribution of PrefLib SOC profiles used in our experiments. Profiles are grouped by the number of alternatives ($m$), and each boxplot shows the distribution of the number of voters ($n$) within that group. Parenthesized values on the $x$-axis give the number of profiles in each bin.}
        \label{fig:preflib-distribution}
    \end{minipage}\hfill%
    \begin{minipage}[t]{0.48\linewidth}
        \centering
        \includegraphics[height=\arxivpairedheight,keepaspectratio]{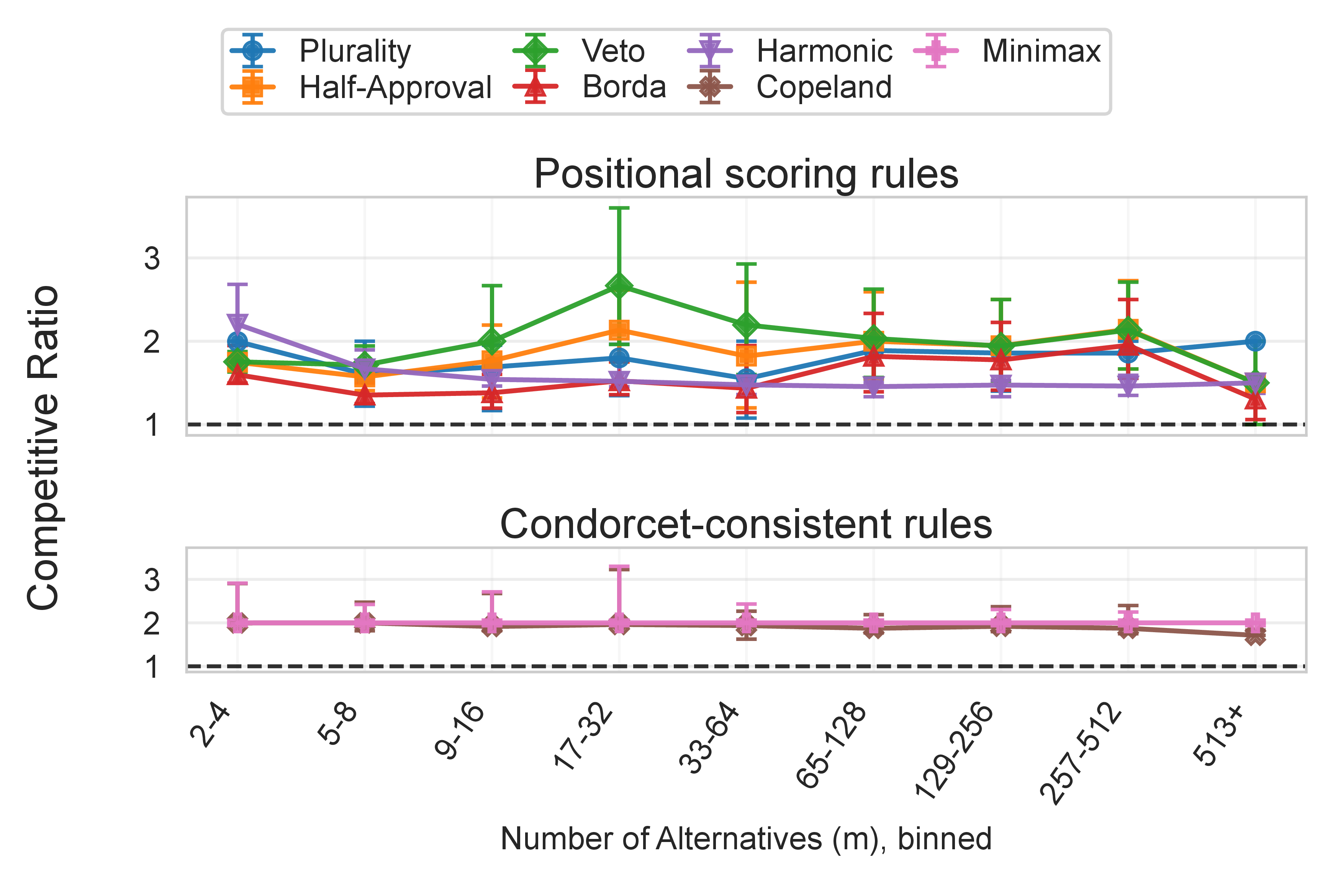}
        \caption{Scalability of the competitive ratio ($\levelpr$ relative to the exact OPT solver) as a function of the number of alternatives ($m$), with $m$ grouped into logarithmic-size bins. Points show the median ratio within each bin and error bars show the interquartile range.}
        \label{fig:plot1_cr}
    \end{minipage}%
\end{figure}

\paragraph{Voting Rules.}
We consider a diverse and representative set of prominent voting rules:
\begin{itemize}
    \item Positional scoring rules (PSRs): \textbf{Borda}, \textbf{Harmonic}, and three instantiations of \textbf{$k$-Approval}, namely \textbf{Plurality} ($k=1$), \textbf{Half-Approval} ($k=\floor{m/2}$), and \textbf{Veto} ($k=m-1$), and
    \item Condorcet-consistent rules: \textbf{Copeland} and \textbf{Minimax}.
\end{itemize}

\paragraph{Baselines for Comparison.}
For these voting rules, we compare the performance of the following three approaches:
\begin{itemize}
    \item \textbf{OPT (Optimal):} The absolute minimum number of queries required to prove the winner. This is computed via an Integer Linear Programming (ILP) formulation provided in \Cref{app:ilp} using the Gurobi solver.
    \item \textbf{$\levelpr$:} Our proposed algorithm (\Cref{alg:level-pruning}) that queries voters level-by-level while pruning those who have revealed all possible-winners.
    \item \textbf{$\level$:} A standard baseline that queries voters level-by-level, without any pruning, until one candidate emerges as the necessary winner. Note that the maximum depth $d^{\max}$ queried under $\level$ and $\levelpr$ is identical, but $\levelpr$ stops querying some voters at depth lower than $d^{\max}$.
\end{itemize}
Since Copeland and Minimax also define their own scores and choose winners that maximize these scores, \levelpr{} remains well-defined for these rules.

To ensure a reasonable timeframe, we impose a $60$-second time limit on the ILP for $\opt$; for each rule, this excluded less than $1\%$ of the instances, leaving at least $7,000$ instances on which $\opt$ (and hence, the competitive ratios of $\level$ and $\levelpr$) could be computed.


\paragraph{Evaluation Metrics.}
We assess the algorithms based on two primary metrics:
\begin{itemize}
    \item \textbf{Competitive Ratio:} The ratio of the queries used by $\levelpr$ or $\level$ to the optimal number of queries $\opt$. This sheds light on the overall performance of the two algorithms.
    \item \textbf{Pearson Correlation:} The Pearson correlation coefficient between the depth to which voter $i$ is queried under $\levelpr$ or $\opt$ and the position of the true winner in voter $i$'s ranking. This sheds light on the probing strategy of an algorithm. Ideally, an algorithm should query each voter approximately up to where she can reveal the true winner, resulting in a high correlation coefficient.
\end{itemize}

\begin{figure}[htb!]
    \centering
    \begin{minipage}[t]{0.48\linewidth}
        \centering
        \includegraphics[width=\linewidth]{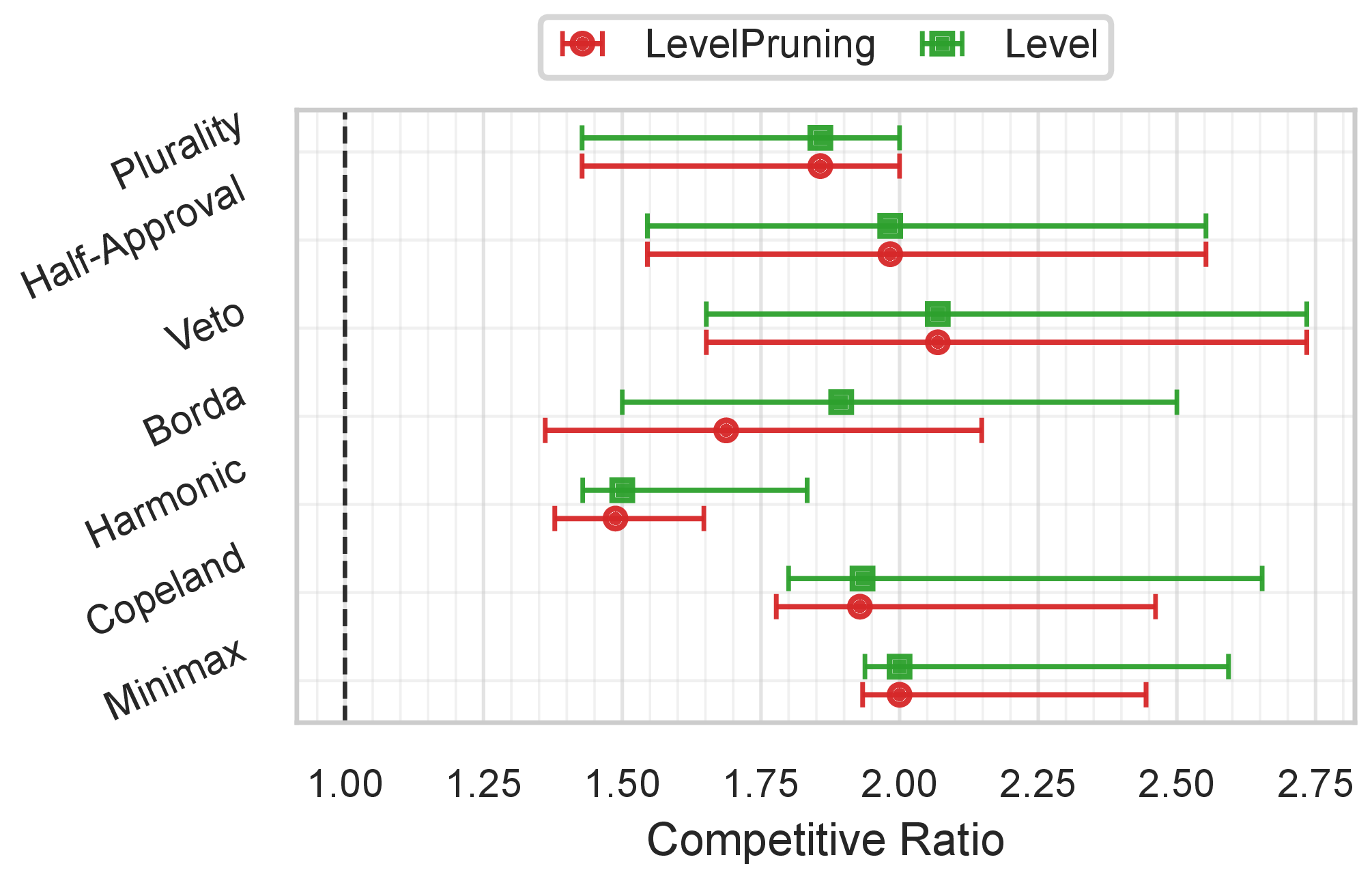}
        \caption{Performance comparison between the proposed $\levelpr$ and the $\level$ baseline. Points show the median competitive ratio across PrefLib instances and horizontal bars show the interquartile range.}
        \label{fig:plot2_bar}
    \end{minipage}\hfill%
    \begin{minipage}[t]{0.48\linewidth}
        \centering
        \includegraphics[width=\linewidth]{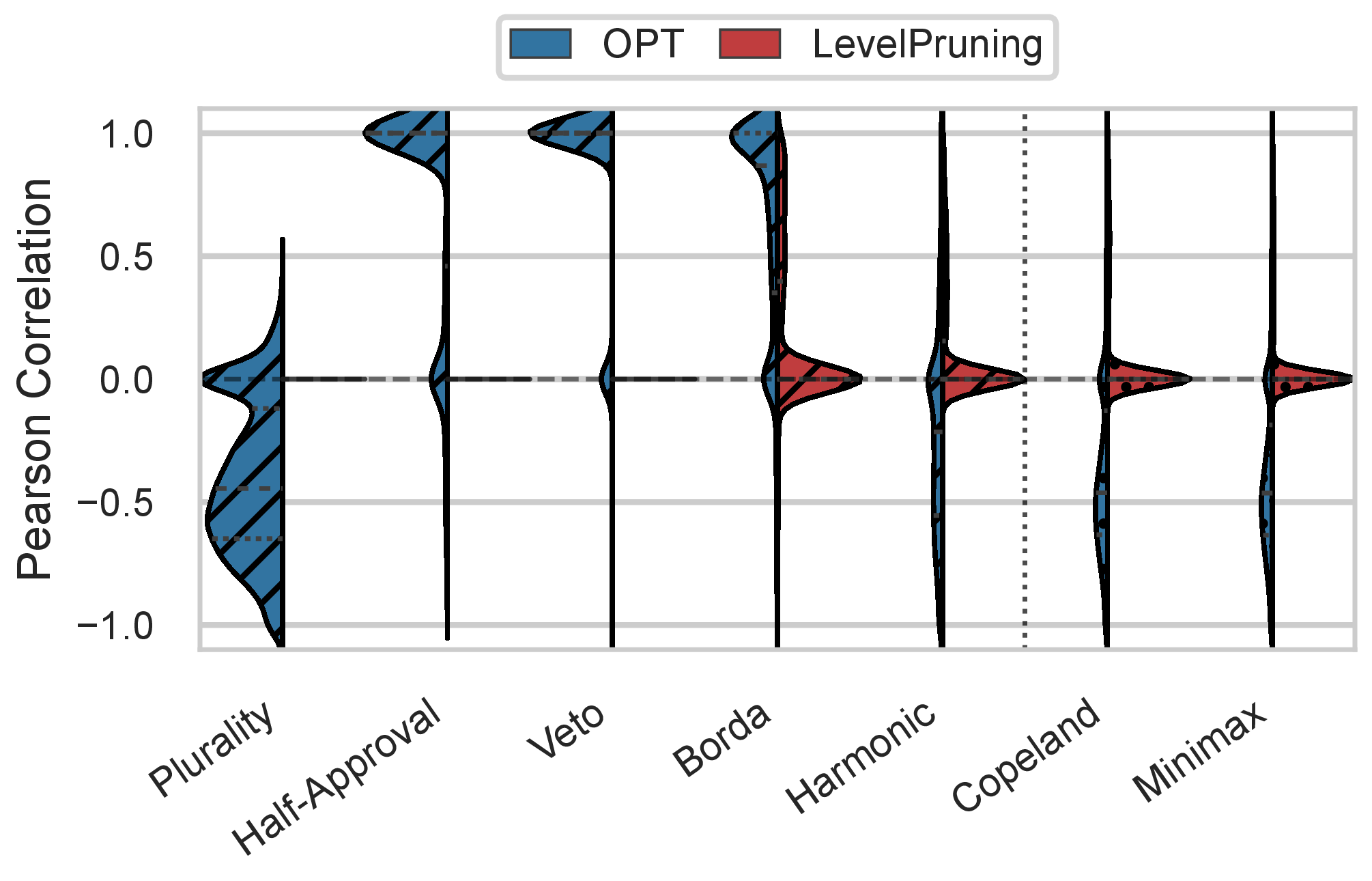}
        \caption{Violin plots illustrating the Pearson correlation between a voter's true rank of the winning candidate and their required query depth under both OPT and $\levelpr$ paradigms. The vertical dashed line separates positional scoring rules on the left from Condorcet-based rules on the right. The distinct distributions reveal the underlying algorithmic strategies for each voting rule.}
        \label{fig:plot3_violin}
    \end{minipage}%
\end{figure}

\subsection{Results and Analysis}
\paragraph{Performance.} \Cref{fig:plot1_cr} shows that the median competitive ratio of $\levelpr$ remains rather small (typically close to $2$) and does not exhibit any clear upward trend as the number of alternatives grows. Its performance is fairly consistent across Condorcet-consistent rules, slightly better on Borda, and slightly worse on Veto. The empirical comparison between Borda and Veto is consistent with their worst-case analyses.

\paragraph{Does Pruning Help?} \Cref{fig:plot2_bar}  compares the competitive ratios of $\levelpr$ and $\level$. For $k$-approval rules (Plurality, Half-Approval, and Veto), they are identical for a simple reason: until all voters are queried for their top $k$ positions (or a unanimous winner with a score of $n$ is observed), no alternative (and, thus, no voter) can be pruned as such an alternative can possibly appear in the $k$-th position of all voters who have not revealed it yet, making it a possible winner. Pruning provides modest savings for Harmonic, Copeland, and Minimax, and noticeably more for Borda.

\paragraph{Probing Strategy.} \Cref{fig:plot3_violin} depicts the correlation between ranks to which voters are queried in $\levelpr$ and $\opt$ and ranks at which they place the true winner. $\levelpr$ maintains a near-zero correlation because it largely queries level-by-level, querying most voters to the same maximum depth $d^{\max}$; indeed, for $k$-approval rules, we observed above that it never prunes any voters, resulting in a Pearson correlation coefficient of exactly $0$ on every instance. The probing strategy of $\opt$, on the other hand, is highly complex, and shows a tension between two competing tendencies:
\begin{itemize}
    \item \textbf{Negative correlation:} For Condorcet-consistent rules such as Copeland and Minimax, the negative correlation comes largely from the fact that $\opt$ chooses to make no queries to many detractors who rank the winner low; instead, it verifies the winner by making a few queries to its supporters, thereby uncovering it at high ranks and deducing that it defeats the large number of alternatives placed below it in those rankings.
    \item \textbf{Positive correlation with peaks at $1$ and $0$:} For $k$-approval rules with a large $k$, $\opt$ has a dichotomous behavior. When a voter ranks the winner high (thus in the top $k$), $\opt$ makes just as many queries as required to uncover the winner, confirming that the winner gets a score of $1$ from that voter in the $k$-approval tally; the number of queries being equal to the rank of the winner creates a peak at $1$. In contrast, when a voter ranks the winner low (but still within the top $k$), $\opt$ queries the voter till the $k$-th position regardless of the rank of the winner as making a few extra queries after uncovering the winner is worth it to also establish that all alternatives who do not appear in the top $k$ do not get a score of $1$ from that voter in the $k$-approval tally; the number of queries being exactly $k$ regardless of the rank of the winner creates a peak around $0$. This pattern shows up for rules such as Half-Approval and Veto.
\end{itemize}
Plurality is special: voters who place the winner at rank $1$ are queried till rank $1$, and other voters who place the winner at rank greater than $1$ are queried till rank either $0$ or $1$, resulting in the largely negative correlation. Borda and Harmonic lie in-between with Borda acting similarly to $k$-approval with a larger $k$ while Harmonic, with its sharply dropping scores, acts more like $k$-approval with a smaller $k$.

\section{Discussion}\label{sec:discussion}

Our work shows that active elicitation of preferences coupled with the use of
competitive ratio to measure performance has remained severely underexplored in
social choice theory---a paradigm whose relevance continues to grow with online platforms deciding between a large number of alternatives by actively eliciting stakeholder preferences.

\paragraph{The optimal competitive ratio for Borda.} An immediate direction for future work is to settle the optimal competitive ratio for Borda count. Our two algorithms exploit different structures and yield complementary results: $\levelpr$ uses level-wise pruning, extends to concave scoring rules, and has an $O(m^{2/3})$ guarantee for Borda, whereas $\multiscale$ uses multi-scale score thresholds and improves the worst-case upper bound for Borda to $O(\sqrt m)$. The $\Omega(\sqrt m)$ construction for $\levelpr$ is algorithm-specific; the best lower bound against arbitrary Borda elicitation algorithms remains the constant $3$. It therefore remains open whether Borda admits a constant-competitive elicitation algorithm, or whether every algorithm must have a competitive ratio that grows with $m$. Extending the theoretical analysis to other scoring rules of interest, such as the harmonic rule, would also be interesting.

\paragraph{Alternative query models.} Our guarantees are specific to the next-best query model, where voters gradually extend a revealed prefix of their preference rankings. Other preference-elicitation interfaces are also common. For example, platforms such as \polis and \remesh often collect approval preferences or pairwise comparisons. Optimizing the elicitation of such preferences via competitive ratio analysis is an interesting future direction.

\paragraph{Extension to other models.}
More broadly, competitive analysis may be useful beyond winner determination under popular rules. For example, there is a rapidly expanding literature on distortion, which studies how limited preference information affects welfare approximation in voting~\cite{AFSV21}. In particular, communication-distortion tradeoffs ask how many bits or queries suffice in the worst case over all instances to achieve a desired approximation guarantee~\cite{MPSW19,MSW20,kempe2020communication} whereas our competitive-ratio perspective asks how many queries are needed in a given instance, relative to the instance-specific optimum, to achieve a given target. In this paper, the target was winner determination under popular voting rules, but future work can replace this by a desired distortion guarantee or a desired level of fairness, e.g., proportional fairness~\cite{ebadian2024optimized} or justified representation~\cite{lackner2023multi}.

\section*{Acknowledgments}
This work was supported by an NSERC Discovery Grant and an NSERC-CSE Research Communities
Grant. Researchers funded through the NSERC-CSE Research Communities Grants do not represent
the Communications Security Establishment Canada or the Government of Canada. Any research,
opinions or positions they produce as part of this initiative do not represent the official views of the
Government of Canada.

\section*{AI Disclosure}
The authors used GPT-5.6-Sol to improve the presentation of several proofs. Further, the $O(\sqrt{m})$ improvement to the competitive ratio for Borda count in \Cref{sec:multiscale-borda} was derived entirely by GPT-5.6-Sol based on the conference version of the paper that contained the $O(m^{2/3})$ bound and its analysis.

\printbibliography

@STRING{proc = {Proceedings of the }}

@STRING{ijcai = { International Joint Conference on Artificial Intelligence (IJCAI)}}

@STRING{aamas = { International Conference on Autonomous Agents and Multi-Agent Systems (AAMAS)}}

@STRING{aaai = { AAAI Conference on Artificial Intelligence (AAAI)}}

@STRING{ec = { ACM Conference on Economics and Computation (EC)}}

@STRING{emnlp = {  Conference on Empirical Methods in Natural Language Processing (EMNLP)}}

@STRING{icml = { International Conference on Machine Learning (ICML)}}

@STRING{nips = { Annual Conference on Neural Information Processing Systems (NeurIPS)}}

@STRING{neurips = { Annual Conference on Neural Information Processing Systems (NeurIPS)}}

@STRING{adt = { International Conference on Algorithmic Decision Theory (ADT)}}

@STRING{springer = {Springer-Verlag}}

@STRING{acm = {ACM Press}}

@STRING{ieee = {IEEE Press}}

@article{ebadian2024optimized,
  title={Optimized distortion and proportional fairness in voting},
  author={Ebadian, Soroush and Kahng, Anson and Peters, Dominik and Shah, Nisarg},
  journal={ACM Transactions on Economics and Computation},
  volume={12},
  number={1},
  pages={1--39},
  year={2024},
  publisher={ACM New York, NY}
}

@inproceedings{borodin2022distortion,
  title={Distortion in voting with top-t preferences},
  author={Borodin, Allan and Halpern, Daniel and Latifian, Mohamad and Shah, Nisarg},
  booktitle= proc # {31st} # ijcai,
  pages={116--122},
  year={2022}
}

@inproceedings{HMSK21,
  title={Necessarily optimal one-sided matchings},
  author={Hosseini, Hadi and Menon, Vijay and Shah, Nisarg and Sikdar, Sujoy},
  booktitle= proc # {35th} # aaai,
  pages={5481--5488},
  year={2021}
}

@inproceedings{AFSV21,
  title={Distortion in Social Choice Problems: The First 15 Years and Beyond},
  author={Anshelevich, Elliot and Filos-Ratsikas, Aris and Shah, Nisarg and Voudouris, Alexandros A},
  booktitle= proc # {30th} # ijcai,
  pages={4294--4301},
  year={2021}
}

@inproceedings{MPSW19,
  author = {D. Mandal and A. D. Procaccia and N. Shah and D. P. Woodruff},
  title = {Efficient and Thrifty Voting by Any Means Necessary},
  booktitle = proc # {33rd} # nips,
  pages = {7178--7189},
  year = 2019
}

@inproceedings{MSW20,
  title={Optimal communication-distortion tradeoff in voting},
  author={Mandal, Debmalya and Shah, Nisarg and Woodruff, David P},
  booktitle= proc # {21st} # ec,
  pages={795--813},
  year={2020}
}

@inproceedings{kempe2020communication,
  title={Communication, distortion, and randomness in metric voting},
  author={Kempe, David},
  booktitle= proc # {34th} # aaai,
  pages={2087--2094},
  year={2020}
}

@inproceedings{ebadian2024metric,
  title={Metric Distortion with Elicited Pairwise Comparisons},
  author={Ebadian, Soroush and Halpern, Daniel and Micha, Evi},
  booktitle= proc # {33rd} # ijcai,
  pages={2791--2798},
  year={2024}
}

@article{halpern2026representation,
  author  = {Halpern, Daniel and Kehne, Gregory and Procaccia, Ariel D. and Tucker-Foltz, Jamie and W\"{u}thrich, Manuel},
  title    = {Representation with incomplete votes},
  journal = {Theory and Decision},
  year    = {2026},
  pages   = {257--296}
}

@inproceedings{oren2014online,
  title={Online (budgeted) social choice},
  author={Oren, Joel and Lucier, Brendan},
  booktitle= proc # {28th} # aaai,
  pages={1456--1462},
  year={2014}
}

@article{chen2018optimal,
  title={Optimal Instance Adaptive Algorithm for the Top-$ K $ Ranking Problem},
  author={Chen, Xi and Gopi, Sivakanth and Mao, Jieming and Schneider, Jon},
  journal={IEEE Transactions on Information Theory},
  volume={64},
  number={9},
  pages={6139--6160},
  year={2018}
}

@inproceedings{peters2022online,
  title={Online elicitation of necessarily optimal matchings},
  author={Peters, Jannik},
  booktitle= proc # {36th} # aaai,
  pages={5164--5172},
  year={2022}
}

@inproceedings{PriceZhao2023,
  title={A competitive algorithm for agnostic active learning},
  author={Price, Eric and Zhou, Yihan},
  booktitle= proc # {37th} # neurips,
  pages={58670--58691},
  year={2023}
}

@inproceedings{zhao2024electoral,
  title={An electoral approach to diversify llm-based multi-agent collective decision-making},
  author={Zhao, Xiutian and Wang, Ke and Peng, Wei},
  booktitle= proc # {2024} # emnlp,
  pages={2712--2727},
  year={2024}
}

@inproceedings{pennock2000normative,
  title={A Normative Examination of Ensemble Learning Algorithms},
  author={Pennock, David M and Maynard-Reid, Pedrito and Giles, C Lee and Horvitz, Eric},
  booktitle= proc # {17th} # icml,
  pages={735--742},
  year={2000}
}

@book{borodin2005online,
  title={Online computation and competitive analysis},
  author={Borodin, Allan and El-Yaniv, Ran},
  year={2005},
  publisher={Cambridge University Press}
}

@article{sleator1985amortized,
  title={Amortized efficiency of list update and paging rules},
  author={Sleator, Daniel D and Tarjan, Robert E},
  journal={Communications of the ACM},
  volume={28},
  number={2},
  pages={202--208},
  year={1985}
}

@inproceedings{conitzer2005communication,
  title={Communication complexity of common voting rules},
  author={Conitzer, Vincent and Sandholm, Tuomas},
  booktitle= proc # {6th} # ec,
  pages={78--87},
  year={2005}
}

@inproceedings{service2012communication,
  title={Communication complexity of approximating voting rules},
  author={Service, Travis C and Adams, Julie A},
  booktitle= proc # {11th} # aamas,
  pages={593--602},
  year={2012}
}

@inproceedings{conitzer2007eliciting,
  title={Eliciting single-peaked preferences using comparison queries},
  author={Conitzer, Vincent},
  booktitle= proc # {6th} # aamas,
  pages={1--8},
  year={2007}
}

@misc{konczak2005voting,
  title        = {Voting Procedures with Incomplete Preferences},
  author       = {Konczak, Kathrin and Lang, J{\'e}r{\^o}me},
  howpublished = {Presented at the IJCAI-05 Multidisciplinary Workshop on Advances in Preference Handling},
  year         = {2005},
  note         = {Unpublished}
}

@article{xia2011determining,
  title={Determining possible and necessary winners given partial orders},
  author={Xia, Lirong and Conitzer, Vincent},
  journal={Journal of Artificial Intelligence Research (JAIR)},
  volume={41},
  pages={25--67},
  year={2011}
}

@article{tessler2024ai,
author = {Michael Henry Tessler  and Michiel A. Bakker  and Daniel Jarrett  and Hannah Sheahan  and Martin J. Chadwick  and Raphael Koster  and Georgina Evans  and Lucy Campbell-Gillingham  and Tantum Collins  and David C. Parkes  and Matthew Botvinick  and Christopher Summerfield },
title = {{AI} can help humans find common ground in democratic deliberation},
journal = {Science},
volume = {386},
number = {6719},
pages = {eadq2852},
year = {2024},
}

@misc{KSIO23,
	author = {A. Konya and L. Schirch and C. Irwin and A. Ovadya},
	title = {Democratic Policy Development Using Collective Dialogues and {AI}},
	howpublished = {arXiv:2311.02242},
	year = 2023
}

@inproceedings{KTAA+25,
  author    = {Konya, Andrew and Thorburn, Luke and Almasri, Wasim and Leshem, Oded Adomi and Procaccia, Ariel D. and Schirch, Lisa and Bakker, Michiel A.},
  title     = {Using collective dialogues and {AI} to find common ground between {Israeli} and {Palestinian} peacebuilders},
  booktitle = {Proceedings of the 2025 ACM Conference on Fairness, Accountability, and Transparency (FAccT)},
  pages     = {312--333},
  year      = {2025}
}

@article{SBES+21,
	author = {C. Small and M. Bjorkegren and T. Erkkil\"a and L. Shaw and C. Megill},
	title = {Polis: Scaling Deliberation by Mapping High Dimensional Opinion Spaces},
	journal = {Revista De Pensament I An\`alisi},
	volume = 26,
	number = 2,
	year = 2021
}

@book{lackner2023multi,
  author = {Lackner, Martin and Skowron, Piotr},
  title = {Multi-Winner Voting with Approval Preferences},
  year = {2023},
  publisher = {Springer}
}

@book{BCEL+16,
	editor = {F. Brandt and V. Conitzer and U. Endriss and J. Lang and A. D. Procaccia},
	title = {Handbook of Computational Social Choice},
	publisher = {Cambridge University Press},
	year = 2016
}

@article{Fish74,
	author = {P. C. Fishburn},
	title = {Paradoxes of Voting},
	journal = {American Political Science Review},
	volume = 68,
	number = 2,
	pages = {537-546},
	year = 1974
}

@inproceedings{MPC13,
	author = {A. Mao and A. D. Procaccia and Y. Chen},
	title = {Better Human Computation Through Principled Voting},
	booktitle = proc # {27th} # aaai,
	pages = {1142--1148},
	year = 2013
}

@inproceedings{MW13,
	author = {N. Mattei and T. Walsh},
	title = {Pref{L}ib: A Library for Preferences},
	booktitle = proc # {3rd} # adt,
	pages = {259-–270},
	year = 2013
}

\newpage
\appendix
\section*{\centering Appendix}
\section{Proofs for Section 3}

\thmkApproval*
\begin{proof}
\myparagraph{Upper bound.} The algorithm makes $nk$ queries. We show that every certificate has cost at least $n$. Because positions below $k$ do not affect $k$-approval, truncate any certificate after depth $k$. Fix a true winner $w$, let $q$ be the certificate cost, let $F$ be the number of voters queried through depth $k$, and let $R_w$ be the number of voters at which the certificate reveals that $w$ is approved. 
A fully queried top-$k$ prefix is said to exclude $a$ if $a$ does not appear among its $k$ revealed alternatives and therefore receives no approval point from that voter.
For every $a\ne w$, let $X_a$ be the number of fully queried top-$k$ prefixes that exclude $a$. The certificate must make $w$ a necessary winner, so $R_w\ge n-X_a$ for every $a\ne w$. Every fully queried prefix excludes exactly $m-k$ alternatives, and therefore
\[
(m-1)(n-R_w)\le \sum_{a\ne w}X_a\le (m-k)F.
\]
At most \(F\) revealed approvals of \(w\) come from the fully queried voters, while every other revealed approval uses at least one of the \(q-kF\) remaining queries. Hence \(R_w\le F+(q-kF)=q-(k-1)F\). Then
\[
(m-1)\bigl(n-q+(k-1)F\bigr) \le (m-1)(n-R_w)\le (m-k)F.
\]
Rearranging yields
\[
(m-1)(n-q)
\le \bigl(m-k-(m-1)(k-1)\bigr)F.
\]
Since $k\ge 2$ and $F\ge 0$, the right is no more than $0$.
If $q<n$, however, 
\[
0<(m-1)(n-q)
\le \bigl(m-k-(m-1)(k-1)\bigr)F
\le 0,
\]
a contradiction. Thus every certificate has cost $q\ge n$, so
$\opt\ge n$. Since the algorithm makes $nk$ queries, its competitive
ratio is at most $nk/n=k$.

\myparagraph{Lower bound.} Fix \(k\) and \(m\), and let \(n\) be sufficiently large. Let \(A=\{a,b,c_1,\ldots,c_{m-2}\}\). Choose the top-\(k\) sets so that \(a\) is excluded from exactly one ballot, \(b\) is excluded from exactly two ballots, and every \(c_j\) is excluded from at least two ballots. Such sets exist for all sufficiently large \(n\): after reserving the required exclusions of \(a\) and \(b\), distribute the remaining exclusions cyclically among \(c_1,\ldots,c_{m-2}\). Rank \(a\) first whenever it is approved, rank \(b\) in position \(k\) whenever it is approved, and fill the remaining positions consistently with these top-\(k\) sets. Thus \(a\) is the unique winner with score \(n-1\), while every other alternative has score at most \(n-2\).

Consider an adversary that answers \(b\) for the first \(n-2\) distinct voters queried in position \(k\), and reveals the two ballots excluding \(b\) only for the final two such voters. Until both exceptional ballots have been queried through depth \(k\), the transcript has a completion in which \(b\) is approved by at least \(n-1\) voters and hence ties or defeats \(a\). Consequently, every online algorithm must query every voter through depth \(k\), at a cost of \(nk\).

In hindsight, query the first position of the \(n-1\) voters who rank \(a\) first. For each alternative other than \(a\), query through depth \(k\) two ballots that exclude it. These queries establish \(\lb(a)=n-1\) and \(\ub(x)\le n-2\) for every \(x\ne a\), and their total cost is at most \(n-1+2k(m-1)\). It follows that
\[
\comp\ge \frac{nk}{n-1+2k(m-1)}\xrightarrow[n\to\infty]{}k.
\]
Together with the upper bound, this proves the theorem.
\end{proof}

\section{Proofs for Section 4}

\thmBordaLower*
\begin{proof}
We prove the lower bound for \(m=3\), with Borda scores \((2,1,0)\). Fix \(t\ge1\) and let \(n=3t\). The base profile \(P^0\) consists of \(2t\) voters with ranking \(a\succ b\succ c\) and \(t\) voters with ranking \(b\succ c\succ a\). Its score vector is \((\score(a),\score(b),\score(c))=(4t,4t,t)\), so its winner set is \(\{a,b\}\). We use two kinds of exceptional profiles. An \(A\)-exception replaces one ballot \(a\succ b\succ c\) by \(a\succ c\succ b\), giving scores \((4t,4t-1,t+1)\). A \(B\)-exception replaces one ballot \(b\succ c\succ a\) by \(b\succ a\succ c\), giving scores \((4t+1,4t,t-1)\). Thus \(a\) is the unique winner in every exceptional profile.

\myparagraph{Adversary.} Fix a deterministic online algorithm \(\calA\). On the first-position query to each new voter, the adversary answers \(b\) for the first \(t\) such voters and \(a\) for the remaining \(2t\). On every second-position query except the last one, it gives the base-profile response: \(c\) after \(b\) and \(b\) after \(a\). As long as some second position remains unqueried, the transcript is consistent both with the base profile \(P^0\) and with an exceptional profile obtained by placing the exception at an unqueried second position. These completions have different winner sets, so [C1] prevents \(\calA\) from terminating. On the final second-position query, the adversary gives the exceptional response, namely \(a\) after \(b\) or \(c\) after \(a\), producing the corresponding exceptional profile \(P^\star\). Therefore \(\calA\) makes all \(3t\) first-position and all \(3t\) second-position queries, for a total cost of \(6t\).

\myparagraph{Certificate for an \(A\)-exception.} Query the first position of all \(2t\) voters who rank \(a\) first, and query the second position of the exceptional voter and of one ordinary \(a\succ b\succ c\) voter. These \(2t+2\) queries give \(\lb(a)=4t\) and \(\ub(b),\ub(c)\le4t-1\), certifying \(a\) as the unique winner.

\myparagraph{Certificate for a \(B\)-exception.} Query the first position of all \(2t\) voters who rank \(a\) first, and query the first two positions of the exceptional \(b\succ a\succ c\) voter. These \(2t+2\) queries give \(\lb(a)=4t+1\), \(\ub(b)\le4t\), and \(\ub(c)\le4t-2\), again certifying \(a\) as the unique winner. Hence \(\opt(P^\star)\le2t+2\) in either case, and
\[
\comp(\calA)\ge \frac{6t}{2t+2}=\frac{3t}{t+1}\xrightarrow[t\to\infty]{}3.
\]
\end{proof}

\subsection{Full Proof for Concave Scoring Rules}
\label{app:concave-levelpr-proof}

The proof below supplies the supporting structural and certificate lemmas and then proves \Cref{thm:concave-levelpr}. It is included from the standalone proof source so that the arXiv version and the independently checkable proof remain synchronized.

\begingroup
\def\ARXIVAPPENDIX{1}
\ifdefined\ARXIVAPPENDIX
\else
\documentclass[11pt]{article}

\usepackage{amsmath, amssymb, amsthm}
\usepackage{mathtools}
\usepackage[hidelinks]{hyperref}
\usepackage[margin=1in]{geometry}

\newtheorem{theorem}{Theorem}
\newtheorem{lemma}[theorem]{Lemma}
\newtheorem{proposition}[theorem]{Proposition}
\newtheorem{corollary}[theorem]{Corollary}
\theoremstyle{remark}
\newtheorem*{remark}{Remark}
\theoremstyle{definition}
\newtheorem{definition}[theorem]{Definition}

\newcommand{\score}{\operatorname{score}}
\newcommand{\lb}{\operatorname{lb}}
\newcommand{\ub}{\operatorname{ub}}
\newcommand{\opt}{\operatorname{OPT}}
\newcommand{\cost}{\operatorname{cost}}
\newcommand{\comp}{\rho}
\newcommand{\levelpr}{\textsc{LevelPruning}}
\newcommand{\myparagraph}[1]{\medskip\noindent\textbf{#1}\enspace}

\begin{document}
\fi

We work in the next-best query model with $n$ voters and a set $A$ of $m \ge 3$ alternatives, and we write $\pi_i(a)$ for the rank of $a$ in the ranking of voter $i$ and $\score(a) = \sum_{i} s_{\pi_i(a)}$ for the score of $a$ under a scoring vector $(s_1,\ldots,s_m)$ with $s_1 \ge \cdots \ge s_m$. An alternative $a$ is \emph{exact} at voter $i$ if $\pi_i(a) \le e_i$, where $e_i$ is the current query depth at $i$, or if $e_i = m-1$, in which case the rank of the unique unrevealed alternative is inferred to be $m$. Given a transcript, $\lb(a)$ and $\ub(a)$ denote the smallest and the largest score of $a$ over all consistent completions, namely
\[
  \lb(a) = \sum_{i : a \text{ exact at } i} s_{\pi_i(a)} + \sum_{i : a \text{ not exact at } i} s_m,
  \qquad
  \ub(a) = \sum_{i : a \text{ exact at } i} s_{\pi_i(a)} + \sum_{i : a \text{ not exact at } i} s_{e_i+1}.
\]
A transcript is a \emph{certificate} if all its completions share the same winner set, $\opt$ is the minimum cost of a certificate, and $\comp$ is the ratio between the cost of the online algorithm and $\opt$. Because the last-ranked alternative is inferred after \(m-1\) queries, we may assume that no certificate queries position \(m\), and hence every certificate depth satisfies \(e_i\le m-1\).

\begin{definition}[Concave scoring rule]
A scoring rule $(s_1,\ldots,s_m)$ with $s_1>s_m$ is \emph{concave} if
\[
  s_j - s_{j+1} \;\le\; s_{j+1} - s_{j+2}
  \qquad\text{for all } j=1,\ldots,m-2.
\]
Its \emph{top gap} is $\gamma := \frac{s_1-s_2}{s_1-s_m}$, and the \emph{deficit} of an alternative $a$ is $\delta(a) := \frac{n s_1 - \score(a)}{n(s_1-s_m)}$.
\end{definition}

\begin{lemma}[Invariance]
\label{concave:lem:invariance}
Let $\sigma>0$, let $\tau\in\mathbb{R}$, and let $s'=\sigma s+\tau$ coordinatewise. Then $s$ and $s'$ induce the same winner set on every profile, the same family of certificates and hence the same value of $\opt$, and the same run of $\levelpr$; moreover $s'$ is concave if and only if $s$ is, and $s$ and $s'$ have the same top gap and the same deficits. Consequently, we may and do assume $s_1=1$ and $s_m=0$ in the sequel, so that $\gamma = s_1-s_2$ and $\score(w)=n(1-\delta)$ for an alternative $w$ of deficit $\delta$. Furthermore, $\delta(w)\in[0,1)$ for every winner $w$.
\end{lemma}

\begin{proof}
Each of $\score(a)$, $\lb(a)$ and $\ub(a)$ is a sum of exactly one entry of the scoring vector per voter, so all three transform as $x\mapsto\sigma x+n\tau$ when $s$ is replaced by $s'$. Every notion above is defined by comparisons among such quantities, which are preserved by an increasing affine map: the winner set by $\score(a)\ge\score(b)$, the certificate condition and $\opt$ by the requirement that the winner set be the same in all completions, and the run of $\levelpr$ by the tests $\ub(a)\ge\max_b\lb(b)$ that determine the set $P$ and hence the pruning. Concavity is preserved because the second differences of $s'$ are $\sigma$ times those of $s$, and $\gamma$ and $\delta(a)$ are preserved because they are ratios of differences of scores, in which the additive term $n\tau$ and the factor $\sigma$ cancel. Finally, $\score(w)\le ns_1$ gives $\delta(w)\ge0$, while a winner has at least the average score, so $\score(w)\ge\frac nm\sum_{j}s_j>ns_m$ and $\delta(w)<1$.
\end{proof}

\begin{remark}
The top gap must be measured relative to $s_1-s_m$ rather than to $s_1$: the quantity $1-s_2/s_1$ is invariant under scaling but not under shifting, so it is not a property of the rule. For instance, $(2,1,0)$ and $(3,2,1)$ are the same rule, yet the former gives $1-s_2/s_1=1/2$ and the latter gives $1/3$, whereas $\gamma=1/2$ for both.
\end{remark}

\begin{lemma}[Properties of concave rules]
\label{concave:lem:concave}
Let $(s_1,\ldots,s_m)$ be a concave scoring rule, normalized as in Lemma~\ref{concave:lem:invariance} so that $s_1=1$ and $s_m=0$. Then (i) $s_j-s_{j+1}\ge\gamma$ for all $j \in \{1,\ldots,m-1\}$; (ii) $s_j\le 1-(j-1)\gamma$ for all $j \in \{1,\ldots,m\}$; (iii) $1-s_{j+1}\le j/(m-1)$ for all $j \in \{0,1,\ldots,m-1\}$; and (iv) $\gamma\le 1/(m-1)$, with equality if and only if the rule is Borda, i.e., $s_j=(m-j)/(m-1)$.
\end{lemma}

\begin{proof}
Write $x_j := s_j-s_{j+1}$ for $j \in \{1,\ldots,m-1\}$, so that $x_1=\gamma$, the sequence $(x_j)_j$ is non-decreasing by concavity, and $\sum_{j=1}^{m-1}x_j=s_1-s_m=1$. Part (i) holds because $x_j\ge x_1=\gamma$, and part (ii) follows by telescoping, since $1-s_j=\sum_{k<j}x_k\ge(j-1)\gamma$. For part (iii), let $A_j:=\frac1j\sum_{k\le j}x_k$. For every $j \le m-2$ we have $x_{j+1}\ge x_k$ for all $k\le j$ and hence $x_{j+1}\ge A_j$, which gives $A_{j+1}=\frac{jA_j+x_{j+1}}{j+1}\ge A_j$. Thus $(A_j)_j$ is non-decreasing, so $A_j\le A_{m-1}=1/(m-1)$, i.e., $1-s_{j+1}=jA_j\le j/(m-1)$; the case $j=0$ is trivial. Part (iv) is part (iii) with $j=1$. Equality means $A_1=A_{m-1}$, which by the monotonicity of $(A_j)_j$ forces $x_{j+1}=A_j=A_{j+1}$ for every $j\le m-2$, so all differences equal $1/(m-1)$, which is Borda; conversely Borda has $\gamma=1/(m-1)$.
\end{proof}

\begin{remark}
Since the $m-1$ consecutive differences sum to $1$ and each is at least $\gamma$ by concavity, Borda is the borderline linear case and \emph{maximises} $\gamma$ among concave scoring rules, every strictly concave rule having $\gamma<1/(m-1)$. Note also that Borda has $\gamma=1/(m-1)$ irrespective of the representation of its scoring vector. The class of concave rules contains Borda and veto ($\gamma=0$ for $m\ge3$), but neither $k$-approval for $2\le k\le m-2$, whose difference sequence $0,\ldots,0,1,0,\ldots,0$ is not monotone, nor the harmonic rule, whose differences are proportional to $1/(j(j+1))$ and hence decreasing. When $\gamma=0$, the full-elicitation branch of the theorem below applies. For veto, the sharper tight ratio $m-1$ follows from its equivalent representation as $(m-1)$-approval.
\end{remark}

\begin{lemma}[Necessary condition for a certificate]
\label{concave:lem:pairwise}
Let $Q$ be a certificate for a profile $P$ and let $w$ be a winner in $P$. Then $\lb(w)\ge\ub(a)$ for every $a\ne w$.
\end{lemma}

\begin{proof}
Suppose $\lb(w)<\ub(a)$ for some $a\ne w$. We exhibit a single completion of $Q$ in which the score of $w$ is $\lb(w)$ and the score of $a$ is $\ub(a)$. Fix a voter $i$ with depth $e_i$. If both $w$ and $a$ are non-exact at $i$, then $e_i\le m-2$, so the two positions $e_i+1$ and $m$ are both open; place $a$ at position $e_i+1$ and $w$ at position $m$, so that the contributions of these two alternatives at voter $i$ are $s_{e_i+1}$ and $s_m=0$, respectively, which are exactly their contributions to $\ub(a)$ and $\lb(w)$. If exactly one of them is non-exact, place it at position $e_i+1$ if it is $a$ and at position $m$ if it is $w$, which again realises those two contributions; if both are exact, their contributions are already realised. The remaining alternatives are placed arbitrarily in the remaining open positions. Summing over voters, this completion has $\score(w)=\lb(w)<\ub(a)=\score(a)$, so $w$ is not a winner in it, contradicting the assumption that $Q$ is a certificate.
\end{proof}

\begin{lemma}[OPT lower bounds for concave scoring rules]
\label{concave:lem:opt}
Let $(s_1,\ldots,s_m)$ be a concave scoring rule, normalized so that $s_1=1$ and $s_m=0$, and let $w$ be a winner of deficit $\delta$, i.e., $\score(w)=n(1-\delta)$. Then every certificate, and in particular the profile-aware optimum, has cost at least $\max\{n(m-1)\delta,\, n/2\}$.
\end{lemma}

\begin{proof}
Let $Q$ be a certificate with depths $(e_i)_{i=1}^{n}$ and cost $|Q|=\sum_i e_i$, and fix any $a\ne w$, which exists since $m\ge3$.

\medskip\noindent\textbf{Bound 1: $\opt\ge n(m-1)\delta$.} We first claim that the contribution of every voter $i$ to $\ub(a)$ is at least $s_{e_i+1}$. Indeed, if $a$ is not exact at $i$, that contribution is $s_{e_i+1}$; if $a$ is exact with $\pi_i(a)\le e_i$, it is $s_{\pi_i(a)}\ge s_{e_i+1}$; and if $a$ is exact by inference, then $e_i=m-1$ and it is $s_m=s_{e_i+1}$. Combining the claim with Lemma~\ref{concave:lem:pairwise} and with $\lb(w)\le\score(w)$ gives
\[
  n(1-\delta) \;=\; \score(w) \;\ge\; \lb(w) \;\ge\; \ub(a) \;\ge\; \sum_{i=1}^{n}s_{e_i+1},
  \qquad\text{that is,}\qquad
  \sum_{i=1}^{n}\bigl(1-s_{e_i+1}\bigr)\;\ge\;n\delta .
\]
The total reduction that the queries must achieve is thus $n\delta$, and by Lemma~\ref{concave:lem:concave}(iii) the $e_i$ queries made to voter $i$ reduce its term by $1-s_{e_i+1}\le e_i/(m-1)$; note that this is a statement about the average of the first $e_i$ differences, and not about any individual difference, which can well exceed $1/(m-1)$ for a concave rule. Summing over voters yields $|Q|=\sum_i e_i\ge n(m-1)\delta$.

\medskip\noindent\textbf{Bound 2: $\opt\ge n/2$.} Suppose $|Q|<n/2$. Then at most $|Q|<n/2$ voters have $e_i\ge1$, and a voter with $e_i=0$ has no exact alternative because $e_i=0<m-1$. Hence $\lb(w)\le|\{i:e_i\ge1\}|\cdot s_1<n/2$, whereas $\ub(a)\ge|\{i:e_i=0\}|\cdot s_1>n/2$, so $\lb(w)<\ub(a)$ and $Q$ is not a certificate by Lemma~\ref{concave:lem:pairwise}.
\end{proof}

We now analyse $\levelpr$, which elicits the profile level by level, maintains the set $P=\{a\in A:\ub(a)\ge\max_{b}\lb(b)\}$ of alternatives not yet ruled out, and stops querying a voter once every alternative of $P$ is exact at that voter. The algorithm uses no knowledge of $\delta$; the depth $d^*$ and the parameter $\lambda$ below are used only in the analysis, which splits the execution into the iterations up to depth $d^*$ (Phase~1) and the subsequent ones (Phase~2). Throughout, $\Lambda:=\max_{b\in A}\lb(b)$, and $P_d$ denotes the set $P$ computed at the end of iteration $d$. Recall that $\lb$ is non-decreasing and $\ub$ is non-increasing as queries accumulate, so $\Lambda$ is non-decreasing and the sets $P_d$ are nested; recall also that a winner $w$ satisfies $\ub(w)\ge\score(w)\ge\score(b)\ge\lb(b)$ for all $b$ and hence belongs to $P_d$ for every $d$.

\ifdefined\ARXIVAPPENDIX
\begin{proof}[Complete proof of \Cref{thm:concave-levelpr}]
\else
\begin{theorem}[Concave scoring rules]
\label{concave:thm:concave}
Let $(s_1,\ldots,s_m)$ be a concave scoring rule with $m\ge3$ and top gap $\gamma\ge0$, and let $w$ be a winner of deficit $\delta$. Then the cost of $\levelpr$ is at most
\[
  \begin{cases}
    n(m-1), & \text{if } \gamma=0 \text{ or } \delta^2>\gamma,\\[2pt]
    \dfrac{7}{4}\,n\,\dfrac{\delta^{2/3}}{\gamma^{4/3}}+3n+m, & \text{if } \gamma>0,\ \gamma^2\le\delta \text{ and } \delta^2\le\gamma,\\[6pt]
    5n+m, & \text{if } \gamma>0 \text{ and } \delta<\gamma^2 .
  \end{cases}
\]
Consequently, when $\gamma>0$, by Lemma~\ref{concave:lem:opt}, on every profile with $n\ge m$,
\[
  \comp \;\le\; \frac{7}{2\,\bigl(2(m-1)\bigr)^{2/3}\gamma^{4/3}}
  \;+\;\frac{1}{\sqrt{\gamma}}\;+\;12
  \;=\;O\!\left(\frac{1}{m^{2/3}\gamma^{4/3}}\right).
\]
In particular, for Borda, where $\gamma=1/(m-1)$, we get $\comp\le\frac{7}{2^{5/3}}(m-1)^{2/3}+\sqrt{m-1}+12=O(m^{2/3})$.
\end{theorem}
\begin{proof}
\fi
The first case of the cost bound holds because $\levelpr$ never queries a voter more than $m-1$ times. Assume therefore that $\gamma>0$ and $\delta^2\le\gamma$, and set
\[
  \lambda:=
  \begin{cases}
    \gamma^{-1},&\delta=0,\\
    \min\bigl\{(\gamma\delta)^{-1/3},\ \gamma^{-1}\bigr\},&\delta>0.
  \end{cases}
\]
Set
\[
  L:=1-\frac1\lambda,
  \qquad
  \beta:=\max\Bigl\{\lambda\delta,\ \frac1\lambda\Bigr\},
  \qquad
  d^*:=\min\Bigl\{\Bigl\lceil\frac{\beta}{\gamma}\Bigr\rceil,\ m-1\Bigr\}.
\]
and write $\alpha:=s_{d^*+1}$, which is well defined since $d^*\le m-1$. Note that $\lambda\ge1$, because $\gamma^{-1}\ge m-1\ge2$ by Lemma~\ref{concave:lem:concave}(iv) and, when \(\delta>0\), because $\gamma\delta\le\gamma\sqrt\gamma<1$. If \(\delta=0\), then \(\lambda\delta=0\); otherwise,
\(\lambda\delta\le(\gamma\delta)^{-1/3}\delta=\delta^{2/3}/\gamma^{1/3}\le1\), where the last inequality uses \(\delta^2\le\gamma\). Thus
\begin{equation}
  \label{concave:eq:lamdelta}
  \lambda\delta\le1.
\end{equation}

\myparagraph{Choice of depth.} We claim that $1-\alpha\ge\beta$, and hence that $\alpha\le 1-1/\lambda=L$. If $d^*=\lceil\beta/\gamma\rceil$, then Lemma~\ref{concave:lem:concave}(i) gives $1-\alpha=1-s_{d^*+1}\ge d^*\gamma\ge\beta$. Otherwise $d^*=m-1$, so $\alpha=s_m=0$ and $1-\alpha=1\ge\beta$, the last inequality holding because $1/\lambda\le1$ and because of \eqref{concave:eq:lamdelta}. This is the point at which the analysis needs $\beta\ge1/\lambda$ and not merely $\beta\ge\lambda\delta$: without the resulting inequality $\alpha\le L$, the $m-d^*$ alternatives that are not exact at a voter, all of which share the upper-bound contribution $\alpha$, could each carry a surplus over the threshold $L$, and the per-voter surplus bound below would fail.

\myparagraph{Phase 1.} The cost of Phase~1 is at most $nd^*\le n\beta/\gamma+n$. Let $U$ be the set of voters at which $w$ is not exact at the end of Phase~1. Since $w\in P_d$ for every $d$, no voter in $U$ is ever pruned during Phase~1, so every $i\in U$ has depth exactly $d^*$ and satisfies $\pi_i(w)\ge d^*+1$, whence $s_{\pi_i(w)}\le\alpha$. Therefore
\[
  n(1-\delta)\;=\;\score(w)\;=\;\sum_{i\notin U}s_{\pi_i(w)}+\sum_{i\in U}s_{\pi_i(w)}
  \;\le\;(n-|U|)\cdot 1+|U|\cdot\alpha\;=\;n-|U|(1-\alpha),
\]
which gives $|U|\le n\delta/(1-\alpha)\le n/\lambda$. Since $\lb(w)=\sum_{i\notin U}s_{\pi_i(w)}=\score(w)-\sum_{i\in U}s_{\pi_i(w)}\ge\score(w)-|U|\alpha$, we obtain
\[
  \lb(w)\;\ge\;n(1-\delta)-\frac{n\delta}{1-\alpha}\,\alpha
  \;=\;n\left(1-\frac{\delta}{1-\alpha}\right)
  \;\ge\;n\left(1-\frac1\lambda\right)\;=\;nL,
\]
using $1-\alpha\ge\beta\ge\lambda\delta$ in the last step.

\myparagraph{Potential.} Define the potential
\[
  \Phi\;:=\;\sum_{a\in A}\bigl(\ub(a)-\Lambda+\gamma\bigr)^+ ,
\]
where $(x)^+:=\max\{x,0\}$. It is non-negative, and it is non-increasing over the execution because $\ub$ is non-increasing and $\Lambda$ is non-decreasing. Taking the potential relative to $\Lambda$ rather than to $\lb(w)$, including $w$ itself in the sum, and adding the cushion $\gamma$ per alternative are what make the charging argument of Phase~2 valid for every query, including queries to a voter whose only remaining non-exact possible winner is $w$ itself, and queries whose target has surplus exactly zero.

We bound $\Phi$ at the end of Phase~1. An alternative $a\notin P_{d^*}$ satisfies $\ub(a)<\Lambda$ and hence contributes at most $\gamma$. Now let $a\in P_{d^*}$ and let $i$ be any voter. If $a$ is not exact at $i$, then $i$ was not pruned before the end of Phase~1, since a voter pruned at the end of an iteration $d\le d^*$ has all alternatives of $P_d\supseteq P_{d^*}$ exact and exactness is permanent; hence $i$ has depth $d^*$ and its contribution to $\ub(a)$ is $\alpha\le L$. If $a$ is exact at $i$, its contribution is $s_{\pi_i(a)}$, and distinct alternatives that are exact at $i$ occupy distinct ranks. Writing $\ub_i(a)$ for the contribution of voter $i$ to $\ub(a)$ and using $\Lambda\ge\lb(w)\ge nL$, we get
\[
  \bigl(\ub(a)-\Lambda+\gamma\bigr)^+
  \;\le\;\Bigl(\sum_{i=1}^{n}\bigl(\ub_i(a)-L\bigr)\Bigr)^{\!+}+\gamma
  \;\le\;\sum_{i=1}^{n}\bigl(\ub_i(a)-L\bigr)^++\gamma,
\]
and summing over $a\in P_{d^*}$ yields $\Phi\le n\Sigma+m\gamma$, where $\Sigma:=\sum_{j=1}^{m}(s_j-L)^+$ and where the additive $m\gamma$ absorbs both the cushions and the alternatives outside $P_{d^*}$. To bound $\Sigma$, set $T:=1/(\lambda\gamma)$ and note that Lemma~\ref{concave:lem:concave}(ii) gives $s_j-L\le 1-(j-1)\gamma-(1-1/\lambda)=\gamma(T-(j-1))$, so $(s_j-L)^+$ is positive only for $j-1<T$ and is then at most $\gamma(T-(j-1))$. With $N:=\lfloor T\rfloor$,
\[
  \Sigma\;\le\;\gamma\sum_{k=0}^{N}(T-k)
  \;=\;\gamma(N+1)\left(T-\frac N2\right)
  \;\le\;\frac{\gamma}{2}\left(T+\frac12\right)^{2}
  \;=\;\frac{1}{2\lambda^2\gamma}+\frac{1}{2\lambda}+\frac{\gamma}{8},
\]
where the second inequality maximises the quadratic $x\mapsto(x+1)(T-x/2)$ over $x\in\mathbb{R}$ at $x=T-\frac12$.

\myparagraph{Phase 2.}
Within each iteration, expose the queries in an arbitrary order and update the bounds after each revealed answer. Consider a query issued in Phase~2 to an active voter $i$ at rank $d+1$, where $d\ge d^*$. Write $\ub^-$, $\Lambda^-$ and $\ub^+$, $\Lambda^+$ for the bounds immediately before and after this query, respectively, and let
\[
  P^{\mathrm{cur}}
  :=\{a\in A:\ub^-(a)\ge\Lambda^-\},
  \qquad
  P_i^{\mathrm{cur}}
  :=\{a\in P^{\mathrm{cur}}:a\text{ is not exact at }i\}.
\]

Suppose first that some $a\in P_i^{\mathrm{cur}}$ is still not exact at $i$ after the query. Before the query, the contribution of $i$ to $\ub(a)$ is $s_{d+1}$; afterward, since $a$ remains non-exact, its contribution is $s_{d+2}$. Hence, by Lemma~\ref{concave:lem:concave}(i),
\[
  \ub^+(a)
  \le
  \ub^-(a)-(s_{d+1}-s_{d+2})
  \le
  \ub^-(a)-\gamma.
\]
Since $a\in P^{\mathrm{cur}}$, we have $\ub^-(a)\ge\Lambda^-$, while $\Lambda^+\ge\Lambda^-$. Therefore, the term of $a$ in $\Phi$ is $\ub^-(a)-\Lambda^-+\gamma$ before the query and at most
\[
  \bigl(\ub^+(a)-\Lambda^++\gamma\bigr)^+
  \le
  \bigl(\ub^-(a)-\gamma-\Lambda^-+\gamma\bigr)^+
  =
  \ub^-(a)-\Lambda^-
\]
after the query. Thus this term, and hence $\Phi$, decreases by at least $\gamma$. No other term increases because upper bounds are non-increasing and $\Lambda$ is non-decreasing.

Otherwise, every alternative of $P_i^{\mathrm{cur}}$ is exact at $i$ after the query. The alternatives of $P^{\mathrm{cur}}\setminus P_i^{\mathrm{cur}}$ were already exact at $i$, so every alternative of $P^{\mathrm{cur}}$ is now exact at $i$. Since upper bounds can only decrease and $\Lambda$ can only increase during the remainder of the iteration, the possible-winner set computed at its end satisfies
\[
  P_{d+1}\subseteq P^{\mathrm{cur}}.
\]
Consequently, every alternative of $P_{d+1}$ is exact at $i$, and voter $i$ is pruned at the end of iteration $d+1$. Each voter is charged at most one query of this kind.

Every Phase~2 query therefore either decreases $\Phi$ by at least $\gamma$ or is charged to a voter that is subsequently pruned. Since $\Phi$ is non-negative and non-increasing, and at most $n$ queries can be charged to voters, the cost of Phase~2 is at most
\[
  \frac{\Phi}{\gamma}+n,
\]
where $\Phi$ is evaluated at the end of Phase~1.


\myparagraph{Total cost.} Combining the three parts and using $\frac{n}{2\lambda\gamma}\le\frac{n}{4\lambda^2\gamma^2}+\frac n4$, which is the inequality of arithmetic and geometric means, the total cost is at most
\[
  \frac{n\beta}{\gamma}+n+\frac{n}{2\lambda^{2}\gamma^{2}}+\frac{n}{2\lambda\gamma}+\frac n8+m+n
  \;\le\;
  \frac{n\beta}{\gamma}+\frac{3n}{4\lambda^{2}\gamma^{2}}+\frac{19n}{8}+m .
\]
If $\delta\ge\gamma^2$, then $\lambda=(\gamma\delta)^{-1/3}$, because $(\gamma\delta)^{-1/3}\le\gamma^{-1}$ is equivalent to $\delta\ge\gamma^{2}$, and $\beta=\lambda\delta$, because $\lambda^{2}\delta=\delta^{1/3}\gamma^{-2/3}\ge1$ is likewise equivalent to $\delta\ge\gamma^{2}$. Substituting, $\frac{n\beta}{\gamma}=n\delta^{2/3}\gamma^{-4/3}$ and $\frac{3n}{4\lambda^{2}\gamma^{2}}=\frac34n(\gamma\delta)^{2/3}\gamma^{-2}=\frac34n\delta^{2/3}\gamma^{-4/3}$, which together with $\frac{19}{8}\le3$ gives the second case of the cost bound. If instead $\delta<\gamma^{2}$, then $\lambda=\gamma^{-1}$ and $\lambda^{2}\delta=\delta\gamma^{-2}<1$, so $\beta=1/\lambda=\gamma$; substituting, $\frac{n\beta}{\gamma}=n$ and $\frac{3n}{4\lambda^{2}\gamma^{2}}=\frac{3n}{4}$, and the total is at most $n+\frac{3n}{4}+\frac{19n}{8}+m\le5n+m$, which is the third case.



\myparagraph{Competitive ratio.}
Assume $n\ge m$ and recall from Lemma~\ref{concave:lem:opt} that
\[
  \opt
  \ge
  n\max\left\{(m-1)\delta,\frac12\right\}.
\]
If $\delta^2>\gamma$, then $\cost(\levelpr)\le n(m-1)$ and
$\opt\ge n(m-1)\delta$, so
\[
  \comp
  \le
  \frac1\delta
  <
  \frac1{\sqrt{\gamma}}.
\]
If $\delta<\gamma^2$, then
\[
  \cost(\levelpr)\le5n+m\le6n,
\]
and $\opt\ge n/2$, so $\comp\le12$.

It remains to consider the regime
\[
  \gamma^2\le\delta
  \qquad\text{and}\qquad
  \delta^2\le\gamma.
\]
Since $n\ge m$, the second cost bound gives
\[
  \cost(\levelpr)
  \le
  C_1+4n,
  \qquad
  C_1:=
  \frac74 n\delta^{2/3}\gamma^{-4/3}.
\]
Using the two lower bounds on $\opt$ separately for $C_1$, and using
$\opt\ge n/2$ for the additive term, we obtain
\[
\begin{aligned}
  \comp
  &\le
  \frac{C_1}{\opt}+\frac{4n}{\opt}\\
  &\le
  \min\left\{
    \frac{C_1}{n/2},
    \frac{C_1}{n(m-1)\delta}
  \right\}+8\\
  &=
  \min\left\{
    \frac72\,\frac{\delta^{2/3}}{\gamma^{4/3}},
    \frac{7}{4(m-1)}
    \frac{1}{\delta^{1/3}\gamma^{4/3}}
  \right\}+8.
\end{aligned}
\]
The first expression inside the minimum is increasing in $\delta$, whereas
the second is decreasing. Over all $\delta>0$, their minimum is maximized
where they coincide, namely at
\[
  \delta=\frac{1}{2(m-1)}.
\]
Enlarging the domain from the present regime to all $\delta>0$ can only
increase this maximum. Therefore,
\[
  \min\left\{
    \frac72\,\frac{\delta^{2/3}}{\gamma^{4/3}},
    \frac{7}{4(m-1)}
    \frac{1}{\delta^{1/3}\gamma^{4/3}}
  \right\}
  \le
  \frac{7}
       {2\bigl(2(m-1)\bigr)^{2/3}\gamma^{4/3}}.
\]
Writing
\[
  F:=
  \frac{7}
       {2\bigl(2(m-1)\bigr)^{2/3}\gamma^{4/3}},
\]
the three regimes give
\[
  \comp
  \le
  \max\left\{
    \gamma^{-1/2},\,
    F+8,\,
    12
  \right\}
  \le
  F+\gamma^{-1/2}+12.
\]
This proves the stated explicit bound.

Finally,
\[
  F
  =
  O\left(\frac{1}{m^{2/3}\gamma^{4/3}}\right),
\]
because $m-1=\Theta(m)$. Moreover, using
$\gamma\le1/(m-1)$ from Lemma~\ref{concave:lem:concave}(iv),
\[
  \frac{\gamma^{-1/2}}
       {m^{-2/3}\gamma^{-4/3}}
  =
  m^{2/3}\gamma^{5/6}
  \le
  \frac{m^{2/3}}{(m-1)^{5/6}}
  =
  O(1).
\]
Also,
\[
  m^{-2/3}\gamma^{-4/3}
  \ge
  m^{-2/3}(m-1)^{4/3}
  =
  \Omega(1),
\]
so the constant term is absorbed as well. Consequently,
\[
  \comp
  =
  O\left(\frac{1}{m^{2/3}\gamma^{4/3}}\right).
\]
For Borda, $\gamma=1/(m-1)$, and substitution gives
\[
  \comp
  \le
  \frac{7}{2^{5/3}}(m-1)^{2/3}
  +\sqrt{m-1}+12
  =
  O(m^{2/3}).
\]

\end{proof}

\begin{remark}
The constant-ratio regime for a concave rule is $\delta<\gamma^{2}$, which for Borda reads $\delta<1/(m-1)^{2}$. In the wider regime $\delta<1/m$ the cost bound above is $\Theta(nm^{4/3}\delta^{2/3})$, which at $\delta$ slightly below $1/m$ is $\Theta(nm^{2/3})$ against $\opt=\Theta(n)$, so no constant bound is available there.
\end{remark}

\ifdefined\ARXIVAPPENDIX
  \let\concaveend 
\else
  \let\concaveend\relax
\fi
\concaveend

\endgroup

\subsection{Proof of the LevelPruning Lower Bound}
\thmLevelprBad*
\begin{proof}
Let $m$ be sufficiently large, set
$r=\lfloor \sqrt{m}/2\rfloor$, and choose $n$ divisible by
$r(4r-1)$. Partition the voters into groups of sizes
\[
n_1=\frac{4r-2}{4r-1}n
\qquad\text{and}\qquad
n_2=\frac{1}{4r-1}n.
\]
There are $r$ challengers $c_1,\ldots,c_r$. Among the $n_1$
voters, these challengers occupy positions $2,\ldots,r+1$,
cyclically and with each challenger occurring equally often at
each position. These voters rank $w$ first. Among the $n_2$
voters, the challengers instead occupy positions $1,\ldots,r$,
again cyclically and with each challenger occurring equally often
at each position, while $w$ is ranked in position $4r$. All
remaining alternatives fill the remaining positions; in
particular, they occur at positions at least $r+2$ among the
$n_1$ voters and at positions at least $r+1$ among the $n_2$
voters.

\myparagraph{True scores.}
The Borda score of $w$ is
\[
\score(w)
=
n_1(m-1)+n_2(m-4r)
=
n(m-2).
\]
Each challenger has average rank $(r+3)/2$ among the $n_1$
voters and average rank $(r+1)/2$ among the $n_2$ voters.
Therefore,
\[
\begin{aligned}
\score(c_i)
&=
n_1\left(m-\frac{r+3}{2}\right)
+
n_2\left(m-\frac{r+1}{2}\right) \\
&=
n\left(m-\frac{r+3}{2}\right)+n_2.
\end{aligned}
\]
Consequently,
\[
\score(w)-\score(c_i)
=
\frac{n(r-1)}{2}-n_2
=
n\left(\frac{r-1}{2}-\frac{1}{4r-1}\right)
>0
\]
for sufficiently large $r$. Every remaining alternative is ranked
at position at least $r+1$ by every voter, and hence has score
strictly below $n(m-2)=\score(w)$. Thus $w$ is the unique winner.

\myparagraph{Cost of $\levelpr$.}
Consider the end of level $d=r+1$. At this point every challenger
is exact, while $w$ remains unrevealed at the $n_2$ voters, and
\[
\begin{aligned}
\ub(c_i)-\lb(w)
&=
n_1\left(m-\frac{r+3}{2}\right)
+
n_2\left(m-\frac{r+1}{2}\right)
-
n_1(m-1) \\
&=
\frac{n(2m-4r^2-3r+1)}{2(4r-1)}
>0.
\end{aligned}
\]
The final inequality follows from $m\ge 4r^2$.

At every earlier level, the challengers have identical bounds by
the cyclic construction, their upper bounds are at least their
true scores, and
\[
\lb(w)=n_1(m-1)
\]
after the first level. Up to level $r$, no remaining alternative
has been revealed, and at level $r+1$ the lower bound of every
remaining alternative is at most
\[
n_2(m-r-1)<\lb(w).
\]
Hence all $r$ challengers remain possible through level $d$.

At the start of level $r+1$, every $n_1$ voter still has a
possible challenger unrevealed, while every $n_2$ voter still has
the possible winner $w$ unrevealed. Thus no voter can have been
pruned before this level. Since pruning is performed only after a
level is completed, $\levelpr$ queries every voter at level
$r+1$. Consequently, by the end of that level it has made
\[
n(r+1)=\Omega(n\sqrt{m})
\]
queries.

\myparagraph{Offline certificate.}
Query the $n_1$ voters through depth $1$ and the $n_2$ voters
through depth $4r$. The cost is
\[
n_1+4rn_2=2n,
\]
and $w$'s score is exact, so
\[
\lb(w)=n(m-2).
\]
For every challenger,
\[
\begin{aligned}
\ub(c_i)
&=
n_1(m-2)
+
n_2\left(m-\frac{r+1}{2}\right) \\
&=
\lb(w)-\frac{n_2(r-3)}{2}
<\lb(w),
\end{aligned}
\]
where the final inequality holds for sufficiently large $r$.

For every remaining alternative, its contribution from each
$n_1$ voter is at most $m-2$. At every $n_2$ voter, its upper-bound
contribution is strictly below $m-2$: if it has been revealed by
depth $4r$, then it occurs at position at least $r+1$; otherwise,
its earliest possible position is $4r+1$. Its total upper bound is
therefore strictly below $\lb(w)$. Thus $\opt\le 2n$, and
\[
\frac{\cost(\levelpr)}{\opt}
\ge
\frac{n(r+1)}{2n}
=
\frac{r+1}{2}
=
\Omega(\sqrt{m}).
\]
\end{proof}

\section{Full Analysis of \texorpdfstring{\multiscale}{MultiScaleScore}}
\label{app:multiscale-proof}

\begingroup
\let\multiscaleoldlabel\label
\let\multiscaleoldref\ref
\renewcommand{\label}[1]{\multiscaleoldlabel{multiscale:#1}}
\renewcommand{\ref}[1]{\multiscaleoldref{multiscale:#1}}
\renewcommand{\N}{V}
\newcommand{\A}{A}
\newcommand{\W}{\mathcal{W}}
\newcommand{\OPT}{\opt}
\newcommand{\Exact}{\mathsf{Exact}}
\newcommand{\MSR}{\multiscale}
\newcommand{\TTest}{\scoretest}
\def\ARXIVAPPENDIX{1}
\ifdefined\ARXIVAPPENDIX
\else
\documentclass[11pt]{article}

\usepackage[margin=1in]{geometry}
\usepackage{amsmath,amssymb,amsthm,mathtools}
\usepackage{listings}
\usepackage{microtype}
\usepackage[hidelinks]{hyperref}

\newtheorem{theorem}{Theorem}
\newtheorem{lemma}[theorem]{Lemma}
\newtheorem{remark}[theorem]{Remark}

\DeclareMathOperator*{\argmax}{arg\,max}
\newcommand{\A}{\mathcal{A}}
\newcommand{\N}{\mathcal{N}}
\newcommand{\W}{\mathcal{W}}
\newcommand{\Comp}{\mathcal{C}}
\newcommand{\OPT}{\operatorname{OPT}}
\newcommand{\cost}{\operatorname{cost}}
\newcommand{\score}{\operatorname{score}}
\newcommand{\lb}{\operatorname{LB}}
\newcommand{\ub}{\operatorname{UB}}
\newcommand{\Exact}{\mathsf{Exact}}
\newcommand{\MSR}{%
    \ifmmode\text{\normalfont\textsc{MultiScaleScore}}%
    \else\textsc{MultiScaleScore}\fi}
\newcommand{\TTest}{%
    \ifmmode\text{\normalfont\textsc{ScoreTest}}%
    \else\textsc{ScoreTest}\fi}

\renewcommand{\lstlistingname}{Algorithm}
\lstdefinestyle{pseudocode}{
    basicstyle=\ttfamily\small,
    columns=fullflexible,
    frame=single,
    numbers=left,
    numberstyle=\scriptsize,
    numbersep=8pt,
    xleftmargin=1.5em,
    framexleftmargin=1.2em,
    mathescape=true,
    showstringspaces=false,
    keepspaces=true,
    morekeywords={procedure,while,do,end,choose,query,update}
}

\title{A Multi-Scale Algorithm for Borda Winner Determination}
\author{}
\date{}

\begin{document}
\maketitle
\fi

\ifdefined\ARXIVAPPENDIX
\subsection{Model and Notation}
\else
\section{Model and notation}
\fi

Let $\N=[n]$ be the set of voters and let $\A$ be a set of $m\ge 2$
candidates.  Each voter $i\in\N$ has a strict ranking of $\A$.  For
$c\in\A$, let
\[
    \pi_i(c)\in[m]
\]
denote the rank of $c$ in voter $i$'s ranking, where rank $1$ is best.
We use the normalized Borda score
\[
    s(r):=\frac{m-r}{m-1},
    \qquad
    \score(c)
    :=\sum_{i\in\N}s(\pi_i(c)).
    \tag{1}
\]
Consequently, the complete Borda winner set is
\[
    \W(P)=\argmax_{c\in\A}\score(c),
    \tag{2}
\]
where $P$ denotes the full preference profile.

\paragraph{Next-best queries.}
At any time, the algorithm has elicited a prefix of depth
$d_i\in\{0,\ldots,m-1\}$ from voter $i$.  A next-best query to voter $i$
reveals the candidate at rank $d_i+1$ and then increments $d_i$ by one.
Once $d_i=m-1$, the unique omitted candidate is inferred to have rank
$m$, so no $m$-th query is needed.

A query transcript $Q$ is a \emph{certificate} for a profile $P$ if every
full profile $P'$ consistent with $Q$ has the same complete winner set:
\[
    \W(P')=\W(P).
\]
Let
\[
    \OPT(P)
    :=
    \min\bigl\{|Q|:Q\text{ is a next-best certificate for }P\bigr\}.
    \tag{3}
\]
The minimization in (3) is profile-aware and need not satisfy the online
constraint.  For an online algorithm $\mathcal{M}$, its competitive ratio
is
\[
    \operatorname{CR}(\mathcal{M})
    :=
    \sup_P
    \frac{\cost_{\mathcal{M}}(P)}{\OPT(P)}.
    \tag{4}
\]

\paragraph{Current score bounds.}
Say that the score contribution of $c$ is \emph{exact at voter $i$}, written
$\Exact_i(c)$, if either $c$ appears in the elicited prefix of voter $i$,
or $d_i=m-1$ and $c$ is the unique omitted candidate whose last-place rank
is inferred.  Define
\[
    \lb_i(c)
    :=
    \begin{cases}
        s(\pi_i(c)), & \Exact_i(c),\\
        0,         & \text{otherwise},
    \end{cases}
    \qquad
    \ub_i(c)
    :=
    \begin{cases}
        s(\pi_i(c)), & \Exact_i(c),\\
        s(d_i+1),  & \text{otherwise},
    \end{cases}
    \tag{5}
\]
and let $\lb(c):=\sum_i\lb_i(c)$ and $\ub(c):=\sum_i\ub_i(c)$.  Then
\[
    \lb(c)\le \score(c)\le \ub(c)
    \qquad\text{for every }c\in\A.
    \tag{6}
\]
Moreover, a next-best query on a voter at which $c$ is not exact has one
of two effects on $\ub(c)$:
\begin{itemize}
    \item if the answer is $c$, then $\ub_i(c)$ is unchanged and becomes
    exact;
    \item otherwise, $\ub_i(c)$ decreases by $1/(m-1)$, possibly becoming
    exact by last-place inference.
\end{itemize}
Thus every $\ub(c)$ is monotone nonincreasing during the execution, and
$\lb(c)=\ub(c)=\score(c)$ once $c$ is exact at every voter.

\ifdefined\ARXIVAPPENDIX
\subsection{The Algorithm}
\else
\section{The algorithm}
\fi

\ifdefined\ARXIVAPPENDIX
\subsubsection{A Score Threshold Test}
\else
\subsection{A score threshold test}
\fi

For a candidate $c$ and a score threshold $T$, the procedure
$\TTest(c,T)$ repeatedly queries a voter at which the score contribution
of $c$ is not yet exact.  It stops once either $\score(c)$ is
exact or $\ub(c)<T$.
All queries update the shared global transcript and hence the score bounds
of every candidate.

\begin{lstlisting}[
    style=pseudocode,
    caption={$\TTest(c,T)$},
    label={alg:threshold-test}
]
procedure ScoreTest($c,T$)
    while $c$ is not exact at every voter and $\ub(c)\ge T$ do
        choose the least-index voter $i$ with $\neg\Exact_i(c)$
        query NextBest($i$)
        update the shared transcript, all depths, and all score bounds
    end while
end procedure
\end{lstlisting}

\begin{lemma}[Score-test guarantee]
\label{lem:test}
Suppose that $\ub(c)=\ub_0(c)\ge T$ when
$\TTest(c,T)$ begins.  Then the test uses at most
\[
    (m-1)(\ub_0(c)-T)+1+n
    \tag{7}
\]
queries.  At termination, either $\score(c)$ is exact or
$\ub(c)<T$.
Furthermore,
\[
    \score(c)\ge T
    \quad\Longrightarrow\quad
    \TTest(c,T)\text{ makes }\score(c)\text{ exact}.
    \tag{8}
\]
\end{lemma}

\begin{proof}
Every answer different from $c$ decreases $\ub(c)$ by $1/(m-1)$.  Hence
there can be at most $(m-1)(\ub_0(c)-T)+1$ such answers before $\ub(c)<T$.
An answer equal to $c$ leaves $\ub(c)$ unchanged, but this can happen at
most once per voter, and therefore at most $n$ times.  This proves the
bound in (7).

If $\score(c)\ge T$, then (6) implies
$\ub(c)\ge\score(c)\ge T$ throughout the test.  The test
therefore cannot stop because $\ub(c)<T$ and must instead make the score of
$c$ exact.  This proves (8).
\end{proof}

\ifdefined\ARXIVAPPENDIX
\subsubsection{The Multi-Scale Procedure}
\else
\subsection{The multi-scale procedure}
\fi

For $m\ge5$, define
\[
    G:=2^{\lfloor\log_2\sqrt m\rfloor}.
    \tag{9}
\]
The algorithm considers dyadic scales $g\in\{1,2,4,\ldots,G\}$.  At scale
$g$, it first extends every prefix to depth $h=2g$, sets the score
threshold
\[
    T_g:=\frac{n(m-g)}{m-1},
    \tag{10}
\]
freezes the set $C_g:=\{c\in\A:\ub(c)\ge T_g\}$, and score-tests every
member of this frozen set.  It returns only
after all candidates in $C_g$ have been processed.

\begin{lstlisting}[
    style=pseudocode,
    caption={$\MSR$},
    label={alg:multiscale}
]
procedure MultiScaleScore
    if $m\le 4$ then
        extend every ballot to depth $m-1$
        return $\argmax_{c\in\A}\score(c)$
    end if

    $G \gets 2^{\lfloor\log_2\sqrt m\rfloor}$
    for $g\in\{1,2,4,\ldots,G\}$ do
        extend every ballot whose current depth is below $2g$ to depth $2g$
        $T_g \gets n(m-g)/(m-1)$
        $C_g \gets \{c\in\A:\ub(c)\ge T_g\}$       // freeze this set

        for $c\in C_g$ in a fixed order do
            if $\ub(c)\ge T_g$ then
                ScoreTest($c,T_g$)
            end if
        end for

        $E_g \gets \{c\in C_g:c\text{ is exact and }\score(c)\ge T_g\}$
        if $E_g\ne\varnothing$ then
            return $\argmax_{c\in E_g}\score(c)$
        end if
    end for

    extend every ballot to depth $m-1$
    return $\argmax_{c\in\A}\score(c)$
end procedure
\end{lstlisting}

The algorithm is online and deterministic: the candidate order and every
tie-breaking choice can be fixed in advance, and every query depends only
on the current transcript.  The uniform extension is always legal.  When
$m\ge5$ and $g\le G$, one has $2g\le m-1$: if $G=2$, this follows from
$m\ge5$; if $G\ge4$, then $2G\le G^2-1\le m-1$.

\ifdefined\ARXIVAPPENDIX
\subsection{Analysis}
\else
\section{Analysis}
\fi

\ifdefined\ARXIVAPPENDIX
\subsubsection{The Packing Bound}
\else
\subsection{The packing bound}
\fi

For an integer $h<m$, define the score upper bound obtained by truncating
every ballot after rank $h$:
\[
    \ub^{(h)}(c)
    :=
    \sum_{i\in\N}s\bigl(\min\{\pi_i(c),h+1\}\bigr).
    \tag{11}
\]
If every current prefix has depth at least $h$, then
\[
    \ub(c)\le \ub^{(h)}(c).
    \tag{12}
\]

\begin{lemma}[Few candidates survive at one scale]
\label{lem:packing}
At scale $g$ of Algorithm~\ref{alg:multiscale},
\[
    |C_g|\le2g-1.
    \tag{13}
\]
\end{lemma}

\begin{proof}
For one voter, the ranks form a permutation of $[m]$, so
\[
    \sum_{c\in\A}
    \left(s\bigl(\min\{\pi_i(c),h+1\}\bigr)-s(h+1)\right)
    =
    \frac{1}{m-1}\sum_{r=1}^{h}(h+1-r)
    =
    \frac{h(h+1)}{2(m-1)}.
\]
Summing over all voters gives the score-surplus identity
\[
    \sum_{c\in\A}
    \left(\ub^{(h)}(c)-ns(h+1)\right)
    =
    \frac{nh(h+1)}{2(m-1)}.
    \tag{14}
\]
At scale $g$, take $h=2g$.  If $c\in C_g$, then by (10)--(12),
\[
    \ub^{(h)}(c)\ge \ub(c)\ge T_g.
\]
Thus $c$ contributes at least
\[
    T_g-ns(2g+1)=\frac{n(g+1)}{m-1}
\]
to the left-hand side of (14).  It follows that
\[
    |C_g|
    \le
    \frac{ng(2g+1)/(m-1)}{n(g+1)/(m-1)}
    =
    2g-1+\frac1{g+1}
    <2g.
\]
Since $|C_g|$ is an integer, $|C_g|\le2g-1$.
\end{proof}

\ifdefined\ARXIVAPPENDIX
\subsubsection{Correctness, Including All Ties}
\else
\subsection{Correctness, including all ties}
\fi

\begin{lemma}[A successful scale determines the complete winner set]
\label{lem:correctness}
If $E_g\ne\varnothing$, the set returned at scale $g$ is the complete
Borda winner set in every full profile consistent with the final
transcript at that scale.
\end{lemma}

\begin{proof}
We first claim that, after all tests at scale $g$ finish,
\[
    c\notin E_g
    \quad\Longrightarrow\quad
    \ub(c)<T_g.
    \tag{15}
\]
If $c\notin C_g$, then $\ub(c)<T_g$ when $C_g$ is frozen, and monotonicity
preserves this inequality.  If $c\in C_g$ but its test is skipped, an
earlier test has already made $\ub(c)<T_g$.  If its test runs, it terminates
either with $\ub(c)<T_g$ or with exact $\score(c)<T_g$, in which
case $\ub(c)=\score(c)<T_g$.  This proves (15).

Now let
\[
    s^\star:=\max_{e\in E_g}\score(e)\ge T_g.
\]
Every $e\in E_g$ has an exact score.  On the other hand, in every
completion consistent with the transcript, (15) gives
\[
    \score(c)\le \ub(c)<T_g\le s^\star
    \qquad\text{for every }c\notin E_g.
\]
Hence no candidate outside $E_g$ can maximize score, while the
maximum-score candidates within $E_g$ are known exactly.  The returned set
is therefore precisely the complete set of Borda winners, including all
ties.
\end{proof}

If no scale succeeds, the fallback reveals every ballot to depth $m-1$
and hence determines every score.  Lemma~\ref{lem:correctness} therefore
proves correctness on every profile.

\ifdefined\ARXIVAPPENDIX
\subsubsection{Lower Bounds on the Profile-Aware Optimum}
\else
\subsection{Lower bounds on the profile-aware optimum}
\fi

For a partial transcript $Q$ with prefix depths
$e_i\in\{0,\ldots,m-1\}$, define the score bounds
\[
    \lb_Q(c)
    :=
    \sum_{i\in\N}
    \begin{cases}
        s(\pi_i(c)),&\Exact_i(c),\\
        0,&\text{otherwise},
    \end{cases}
    \tag{16}
\]
and
\[
    \ub_Q(c)
    :=
    \sum_{i\in\N}
    \begin{cases}
        s(\pi_i(c)),&\Exact_i(c),\\
        s(e_i+1),&\text{otherwise}.
    \end{cases}
    \tag{17}
\]
Here $\Exact_i(c)$ is interpreted with respect to $Q$.

\begin{lemma}[Pairwise completion]
\label{lem:pairwise}
A candidate $w$ is a winner in every completion of $Q$ if and only if
\[
    \lb_Q(w)\ge \ub_Q(a)
    \qquad\text{for every }a\ne w.
    \tag{18}
\]
\end{lemma}

\begin{proof}
The ``if'' direction follows directly from the definitions of the score
bounds.  For the converse, suppose that
$\lb_Q(w)<\ub_Q(a)$ for some $a\ne w$.  Complete each ballot so that $w$
attains its minimum possible contribution and $a$ attains its maximum
possible contribution.  If both are unelicited, place $a$ in the first
open position and $w$ in the last open position.  These two positions are
distinct because two distinct unelicited candidates imply that at least
two positions remain open.  If only one of $a,w$ is unelicited, place that
candidate in its corresponding extremal open position; if both are
already exact, their positions are fixed.  The remaining candidates can
be placed arbitrarily.

This gives a single completion in which
\[
    \score(w)=\lb_Q(w)
    <
    \ub_Q(a)=\score(a).
\]
Thus $w$ is not a winner in that completion.
\end{proof}

\begin{lemma}[Certificate lower bounds]
\label{lem:opt}
Let $w$ be any Borda winner and define its normalized score deficit by
\[
    \delta:=1-\frac{\score(w)}n.
    \tag{19}
\]
Then
\[
    \OPT(P)
    \ge
    \max\left\{
        \frac{n(m-1)}m,\,
        n(m-1)\delta
    \right\}.
    \tag{20}
\]
\end{lemma}

\begin{proof}
Consider any certificate for $P$, let $e_i$ be its prefix depths, and put
\[
    q:=\sum_{i\in\N}e_i,
    \qquad M:=m-1,
    \qquad F:=\sum_{i\in\N}(M-e_i)=nM-q.
    \tag{21}
\]
Because the transcript determines the complete winner set, every true
winner $w$ is a winner in every consistent completion.  Hence
Lemma~\ref{lem:pairwise} applies.

For every candidate $a$ and every voter $i$, its normalized score
upper-bound contribution is at least $(M-e_i)/M$.  Indeed, if $a$ is
unelicited it can be placed in the next open position, giving exactly
$(M-e_i)/M$ points; if it is
already revealed at rank $r\le e_i$, then
\[
    s(r)=\frac{m-r}{M}\ge\frac{M-e_i}{M}.
\]
Therefore $\ub_Q(a)\ge F/M$ for every challenger $a\ne w$, and
Lemma~\ref{lem:pairwise} gives
\[
    \lb_Q(w)\ge \frac{F}{M}=n-\frac qM.
    \tag{22}
\]

The true score of $w$ is
\[
    \score(w)=n(1-\delta).
\]
Since $\lb_Q(w)\le\score(w)$, (22) implies
\[
    q\ge nM\delta=n(m-1)\delta.
    \tag{23}
\]

For the other bound, only a voter whose prefix reveals $w$ can contribute
positively to $\lb_Q(w)$.  There are at most $q$ such voters, and each
contribution is at most $1$.  A last-place rank inferred at depth $m-1$
contributes zero and causes no exception.  Thus
\[
    \lb_Q(w)\le q.
\]
Combining this with (22) yields
\[
    q\ge n-\frac qM,
    \qquad\text{so}\qquad
    q\ge\frac{nM}{M+1}=\frac{n(m-1)}m.
    \tag{24}
\]
Both inequalities hold for every certificate, and hence for the
profile-aware optimum.
\end{proof}

\ifdefined\ARXIVAPPENDIX
\subsubsection{Query Cost}
\else
\subsection{Query cost}
\fi

\begin{lemma}[Cost through one dyadic scale]
\label{lem:cost}
If Algorithm~\ref{alg:multiscale} reaches a dyadic scale $g$, then the
total number of queries made up to and including that scale is at most
\[
    5ng^2.
    \tag{25}
\]
\end{lemma}

\begin{proof}
At a dyadic scale $x\le g$, every score test starts with $\ub_0(c)\le n$,
and hence Lemma~\ref{lem:test} bounds its cost by
\[
    (m-1)\left(n-\frac{n(m-x)}{m-1}\right)+1+n=nx+1.
\]
Together with Lemma~\ref{lem:packing}, the total score-test cost is at
most
\[
    |C_x|(nx+1)
    \le(2x-1)(nx+1).
    \tag{26}
\]
Across all uniform-extension steps through scale $g$, at most $2ng$
queries are made: each ballot is merely brought up to depth $2g$, and
queries already made by score tests are never repeated.  Therefore
\[
    Q_{\le g}
    \le
    2ng+
    \sum_{\substack{x\le g\\x\text{ dyadic}}}
        (2x-1)(nx+1).
    \tag{27}
\]
Using
\[
    \sum_{\substack{x\le g\\x\text{ dyadic}}}x=2g-1,
    \qquad
    \sum_{\substack{x\le g\\x\text{ dyadic}}}x^2
    =\frac{4g^2-1}{3},
\]
the right-hand side of (27) is at most $5ng^2$ for every $n\ge1$ and
$g\ge1$.
\end{proof}

\ifdefined\ARXIVAPPENDIX
\begin{proof}[Complete proof of \Cref{thm:multiscale-borda}]
\else
\begin{theorem}[Multi-scale Borda elicitation]
\label{thm:main}
$\MSR$ is a deterministic online next-best elicitation algorithm that
determines the complete Borda winner set.  Moreover,
\[
    \operatorname{CR}(\MSR)
    \le
    \begin{cases}
        4,&2\le m\le4,\\[2mm]
        20\sqrt m,&m\ge5.
    \end{cases}
    \tag{28}
\]
In particular,
\[
    \operatorname{CR}(\MSR)=O(\sqrt m).
\]
\end{theorem}
\begin{proof}
\fi
The construction is online, and Lemma~\ref{lem:correctness} proves correctness.
It remains to prove (28).

First suppose $2\le m\le4$.  Full elicitation costs $n(m-1)$, while
Lemma~\ref{lem:opt} gives
\[
    \OPT(P)\ge\frac{n(m-1)}m.
\]
Thus the ratio is at most $m\le4$.

Now suppose $m\ge5$, and let $G$ be as in (9).  Fix a true Borda winner
$w$, let $\delta$ be its normalized score deficit from (19), and define
the associated score-deficit scale
\[
    \rho:=1+(m-1)\delta.
\]
By the definitions of $\rho$ and $T_g$,
\[
    \score(w)\ge T_g
    \quad\Longleftrightarrow\quad
    g\ge\rho.
\]

\paragraph{Case 1: $\rho\le G$.}
Let $g$ be the first scheduled dyadic scale satisfying $g\ge\rho$.  Then
$\score(w)\ge T_g$, so the winner belongs to $C_g$.
Lemma~\ref{lem:test} makes its score exact, so $E_g\ne\varnothing$ and
the algorithm terminates at this scale.
By Lemma~\ref{lem:cost}, its cost is at most $5ng^2$.

If $g\in\{1,2\}$, this cost is at most $20n$, whereas
Lemma~\ref{lem:opt} gives
$\OPT(P)\ge n(m-1)/m\ge4n/5$.  The ratio is therefore at most $25$,
which is less than $20\sqrt m$ for $m\ge5$.

If $g\ge4$, the minimality of $g$ gives $\rho>g/2$.  Hence
\[
    \OPT(P)
    \ge n(m-1)\delta
    =n(\rho-1)
    >
    n\left(\frac g2-1\right)
    \ge\frac{ng}{4}.
\]
Consequently,
\[
    \frac{\cost_{\MSR}(P)}{\OPT(P)}
    \le
    \frac{5ng^2}{ng/4}
    =20g
    \le20\sqrt m.
    \tag{29}
\]

\paragraph{Case 2: $\rho>G$.}
The complete execution makes at most $n(m-1)$ distinct queries, including
the fallback.  By Lemma~\ref{lem:opt},
\[
    \frac{\cost_{\MSR}(P)}{\OPT(P)}
    <
    \frac{n(m-1)}{n(G-1)}
    =
    \frac{m-1}{G-1}.
    \tag{30}
\]
For $m\ge5$, one has $G\ge2$, so $G-1\ge G/2$.  Since $G$ is the largest
power of two not exceeding $\sqrt m$, one also has
$G>\sqrt m/2$.  Therefore the right-hand side of (30) is less than
$4\sqrt m$, completing the proof.
\end{proof}

\begin{remark}[Why all candidates in $C_g$ must be processed]
It is unsafe to return after the first test that finds a candidate with
exact score at least $T_g$.  Another member of $C_g$ may have a larger
score or may tie the winner.  Freezing $C_g$ and processing every
member is what turns the one-sided score tests into a certificate for
the complete winner set.
\end{remark}

\ifdefined\ARXIVAPPENDIX
    \let\multiscaleend 
\else
    \let\multiscaleend\relax
\fi
\multiscaleend

\endgroup

\section{ILP Formulations for Computing \texorpdfstring{$\opt$}{OPT}}
\label{app:ilp}

This appendix gives the integer linear programs used to compute the exact
offline optimum in our experiments.  These formulations are implemented in
\texttt{next\_best\_query\_tied.py}.  Throughout, let $V$ be the set of
voters, $A$ the set of alternatives, $n=|V|$, and $m=|A|$.  For
voter $i$, let $\pi_i(c)\in\{1,\ldots,m\}$ denote the true rank of alternative
$c$, with rank $1$ being most preferred.  The ILPs use binary variables
$z_{i,k}$, where $z_{i,k}=1$ means that voter $i$ is queried to depth at least
$k$.  The telescoping constraints
\[
z_{i,k}\ge z_{i,k+1}\qquad\forall i\in V,\ k=1,\ldots,m-1
\]
ensure that queries reveal prefixes.  In all formulations the objective is
\[
\min \sum_{i\in V}\sum_{k=1}^{m} z_{i,k}.
\]
If the true voting rule has multiple winners, the ILP is allowed to certify any
one of the tied winners; this matches the experimental definition of the
minimum number of queries needed to prove a correct winner.

\subsection{Positional Scoring Rules}

This formulation is used for Borda, Harmonic, Plurality, Half-Approval, and
Veto by changing only the scoring vector
$\alpha=(\alpha_1,\ldots,\alpha_m)$, where
$\alpha_1\ge\cdots\ge\alpha_m$.  Let $\mathcal{W}\subseteq A$ be the set of true
winners.  For every $a\in\mathcal{W}$, introduce a binary variable $w_a$ indicating
which tied winner is certified, and impose
\[
\sum_{a\in\mathcal{W}} w_a = 1.
\]
Here $w_a$ does not indicate whether $a$ is the unique true winner; all alternatives
in $\mathcal{W}$ are true winners by definition. Instead, $w_a=1$ means
that the ILP chooses $a$ as the single winner to certify. Hence
$\sum_{a\in\mathcal{W}} w_a=1$ makes the formulation choose exactly one
certification target among the true tied winners.
For each alternative $c$, define lower and upper score bounds:
\begin{align}
L_c
&=
\sum_{i\in V}
\left[
\alpha_m+
\bigl(\alpha_{\pi_i(c)}-\alpha_m\bigr)z_{i,\pi_i(c)}
\right], \label{eq:psr-lb}\\
U_c
&=
\sum_{i\in V}
\left[
\alpha_1-
\sum_{k=1}^{\pi_i(c)-1}
\bigl(\alpha_k-\alpha_{k+1}\bigr)z_{i,k}
\right]. \label{eq:psr-ub}
\end{align}
The lower bound gives candidate $c$ its true score from voter $i$ only if $c$
has been revealed; otherwise it assigns the worst possible score $\alpha_m$.
The upper bound starts from the best possible score $\alpha_1$ and decreases
whenever revealed alternatives above $c$ rule out higher positions for $c$.
The winner-certification constraints are
\[
w_a=1 \quad\Longrightarrow\quad L_a\ge U_b
\qquad
\forall a\in\mathcal{W},\ \forall b\in A\setminus\{a\}.
\]
Thus, if $a$ is selected as the certified winner, even its worst-case score
under the revealed prefixes is at least every other alternative's best-case
score. Thus, the ILP certifies one selected member of the true winner set,
rather than necessarily certifying the entire tied winner set.

\subsection{Copeland}

For Copeland, the implementation solves one ILP for each tied true winner
$a\in\mathcal{W}$ and returns the minimum-cost solution among them.  Fix such a
target winner $a$.  For two alternatives $x,y$, let
\[
P_{xy}=\{i\in V:\pi_i(x)<\pi_i(y)\}
\]
be the voters who truly prefer $x$ to $y$, and define the number of revealed
witnesses for $x\succ y$ as
\[
V_{xy}=\sum_{i\in P_{xy}} z_{i,\pi_i(x)}.
\]
Let $\tau_+=\lfloor n/2\rfloor+1$ be the strict-majority threshold.  If $n$ is
even, let $\tau_0=n/2$ be the tie threshold; if $n$ is odd, set
$\tau_0=\tau_+$.  For each ordered pair $(x,y)$ used by the formulation,
introduce binary variables $p_{xy}$ and $t_{xy}$ with
\[
V_{xy}\ge \tau_+ p_{xy},
\qquad
V_{xy}\ge \tau_0 t_{xy}.
\]
When $n$ is odd, the code sets $t_{xy}=p_{xy}$.  The certified pairwise
Copeland contribution of $x$ against $y$ is
\[
s_{xy}=p_{xy}+t_{xy}-1.
\]
Hence $s_{xy}=1$ certifies a pairwise win for $x$, $s_{xy}=0$ certifies a
pairwise tie when ties are possible, and $s_{xy}=-1$ leaves open the possibility
that $y$ beats $x$.
The constraints defining $p_{xy}$ and $t_{xy}$ are one-directional: the ILP may
set these variables to $1$ only when enough witnesses have been revealed, but it
is not forced to set them to $1$ whenever this is possible.  This relaxation is
sound for certification, because increasing $s_{xy}$ can only strengthen the
certificate.  It improves the lower bound on the selected winner's Copeland
score, or decreases the upper bound on a challenger's Copeland score, and hence
can never make an invalid certificate feasible.

The target winner's lower Copeland score is
\[
L^{\mathrm{Cop}}_a=\sum_{x\in A\setminus\{a\}} s_{ax}.
\]
For any challenger $c$, an upper bound on $c$'s Copeland score is
\[
U^{\mathrm{Cop}}_c
=
\sum_{y\in A\setminus\{c\}} -s_{yc}.
\]
The certification constraints are
\[
L^{\mathrm{Cop}}_a \ge U^{\mathrm{Cop}}_c
\qquad
\forall c\in A\setminus\{a\}.
\]
These constraints require the revealed prefixes to certify that the target
winner's worst-case Copeland score is at least every challenger's best-case
Copeland score.

\subsection{Minimax}

Equivalently, since preferences are strict and complete, maximizing
\[
\min_{x\neq c}(n_{c\succ x}-n_{x\succ c})
\]
is the same as minimizing
\[
\max_{x\neq c} n_{x\succ c}.
\]
We use the latter defeat-count formulation in the ILP, so lower values are
better.  The implementation uses a single ILP over all tied true winners.  Let
$\mathcal{W}$ be the true Minimax winner set.  For each $a\in\mathcal{W}$,
the variable $w_a\in\{0,1\}$ indicates whether $a$ is selected as the Minimax
winner to certify.  The constraint
\[
\sum_{a\in\mathcal{W}} w_a=1.
\]
ensures that exactly one true Minimax winner is selected as the certification
target.
For each ordered pair $(x,c)$ with $x\neq c$, define
\[
P_{xc}=\{i\in V:\pi_i(x)<\pi_i(c)\}.
\]
The lower bound on the number of voters certified to prefer $x$ over $c$ is
\[
L_{xc}
=
\sum_{i\in P_{xc}} z_{i,\pi_i(x)}.
\]
The lower bound $L_{xc}$ counts voters whose revealed prefixes already certify
$x\succ_i c$. For a voter with $x\succ_i c$, revealing $x$ is enough to certify
this comparison, since all unrevealed alternatives must lie below the revealed
prefix.
The upper bound on the number of voters who may prefer $x$ over $c$ is
\[
U_{xc}
=
|P_{xc}|
+
\sum_{i\in P_{cx}}\bigl(1-z_{i,\pi_i(c)}\bigr).
\]
The upper bound $U_{xc}$ counts the largest possible number of voters who could
still prefer $x$ to $c$ under some completion consistent with the revealed
prefixes. All voters in $P_{xc}$ can contribute to this upper bound. For voters
in $P_{cx}$, the true order is $c\succ_i x$; such a voter can be ruled out as
supporting $x\succ_i c$ only once $c$ has been revealed, which explains the
term $1-z_{i,\pi_i(c)}$.

For each candidate $c$, introduce an integer variable $M^U_c$ satisfying
\[
M^U_c \ge U_{xc}
\qquad
\forall x\in A\setminus\{c\}.
\]
The variable $M^U_c$ upper-bounds $c$'s maximum pairwise defeat over all
consistent completions. Since lower defeat-count scores are better under the
Minimax formulation used here, this is the quantity we need to upper-bound for
the selected winner.

To compare a selected target winner against a challenger $b$, the ILP also
chooses one opponent witnessing the challenger's lower-bound maximum defeat.
Introduce binary variables $q_{b,y}$ for $y\neq b$ and impose
\[
\sum_{y\in A\setminus\{b\}} q_{b,y}=1.
\]
The variables $q_{b,y}$ select, for each challenger $b$, one opponent $y$ that
witnesses a lower bound on $b$'s maximum defeat. Since $b$'s Minimax defeat is
a maximum over opponents, one such witness is sufficient.

With $M_{\mathrm{big}}=n+1$, the certification constraints are
\[
M^U_a - L_{yb}
\le
M_{\mathrm{big}}\bigl(2-w_a-q_{b,y}\bigr)
\]
for all $a\in\mathcal{W}$, $b\in A\setminus\{a\}$, and
$y\in A\setminus\{b\}$.  When $w_a=1$ and $q_{b,y}=1$, the big-$M$
constraint reduces to
\[
M^U_a \le L_{yb}.
\]
Thus, even the worst-case maximum defeat of the selected winner $a$ is no
larger than a certified lower bound on challenger $b$'s maximum defeat.
Therefore $b$ cannot have a strictly better Minimax defeat-count score than $a$
in any completion consistent with the revealed prefixes.

In short, the ILP certifies a selected true Minimax winner $a$ by upper-bounding
$a$'s maximum defeat and, for every challenger $b$, finding one certified
pairwise defeat of $b$ that is at least as large.

\end{document}